\documentclass[conference]{IEEEtran}
\IEEEoverridecommandlockouts
\usepackage{cite}
\usepackage{amsmath,amssymb,amsfonts,amsthm}
\usepackage{algorithmic}
\usepackage{graphicx}
\usepackage{textcomp}
\usepackage{xcolor}
\usepackage{multirow}
\usepackage{epstopdf}
\usepackage{graphicx}
\usepackage{algorithm}
\usepackage{url}
\usepackage{geometry}
\usepackage{subfigure}
\usepackage{flushend}

\def\BibTeX{{\rm B\kern-.05em{\sc i\kern-.025em b}\kern-.08em
    T\kern-.1667em\lower.7ex\hbox{E}\kern-.125emX}}
\begin{document}

\title{An Iterative Scheme for \\
Leverage-based Approximate Aggregation
}

\author{
\IEEEauthorblockN{1\textsuperscript{st} Shanshan Han}
\IEEEauthorblockA{\textit{School of Computer Science} \\
\textit{Harbin Institute of Technology}\\
Harbin, China \\
sshan0731@hotmail.com}
\and
\IEEEauthorblockN{2\textsuperscript{nd} Hongzhi Wang}
\IEEEauthorblockA{\textit{School of Computer Science} \\
\textit{Harbin Institute of Technology}\\
Harbin, China \\
wangzh@hit.edu.cn}
\and
\IEEEauthorblockN{3\textsuperscript{rd} Jialin Wan}
\IEEEauthorblockA{\textit{School of Computer Science} \\
\textit{Harbin Institute of Technology}\\
Harbin, China \\
wanjialinhit@126.com}
\and
\IEEEauthorblockN{4\textsuperscript{th} Jianzhong Li}
\IEEEauthorblockA{\textit{School of Computer Science} \\
\textit{Harbin Institute of Technology}\\
Harbin, China \\
lijzh@hit.edu.cn}

}
\maketitle
\begin{abstract}
The current data explosion poses great challenges to approximate aggregation with high efficiency and accuracy.
To address this problem, we propose a novel approach to calculate the aggregation answers with a high accuracy using only a small portion of the data.
We introduce leverages to reflect individual differences in the data from a statistical perspective. Two kinds of estimators, the leverage-based estimator, and the sketch estimator (a ``rough picture'' of the aggregation answer), are in constraint relations and iteratively improved according to the actual conditions until their difference is below a threshold.
Due to the iteration mechanism and the leverages, our approach achieves a high accuracy.
Moreover, some features, such as not requiring recording the sampled data and easy to extend to various execution modes, such as the online mode, make our approach well suited to deal with big data.
Experiments show that our approach has an extraordinary performance, and when compared with the uniform sampling, our approach can achieve high-quality answers with only 1/3 sample size.
\end{abstract}

\newtheorem{theorem}{Theorem}
\newtheorem{definition}{Definition}
\newtheorem{example}{Example}
\vspace{-0.5em}
\section{Introduction}\label{sec:introduction}
The development of intelligent devices and informatization has brought about an unprecedented data explosion, which brings great challenges to data aggregation.
When dealing with big data, usually it is impractical to compute an accurate answer by a full scan of the data sets due to the high computation cost, while an approximate aggregation is more economical.
Meanwhile, users today often expect high-quality answers but do not want to wait too long. They also would like the data analysis system to be flexible and easy to extend.
In light of this situation, an efficient, high-precision, and flexible approximate aggregation approach is in great demand.


To effectively execute approximate aggregation on big data and balance the accuracy and efficiency, researchers proposed bi-level sampling~\cite{Haas2004A}, which considers the local variance of the data when generating the sampling rate.
However, it does not consider the individual differences in the data, while data with different features contribute differently to the aggregation answers. For example, in \textsf{SUM} aggregation, some data (outliers) are much too large but count for limited proportions, thus can hardly be sampled.
However, once they are sampled, due to their extremely high values, significant effects occur about the aggregation answers. Under this condition, these data should not be handled identically with others, while neglecting their individual differences produces a loss of accuracy.

To solve this problem, researchers introduced leverages to reflect the different influences of data on the global answers~\cite{Ma2013A}.
The leverage of data is calculated using the data value as well as all the other data. To reflect the individual differences in the data, a biased sampling process is performed, and for each data point, its biased sampling probability is generated using its leverage and the uniform sampling probability.
This technique provides an unbiased estimation of the accurate value~\cite{Ma2013A}. It also considers the individual differences in the data, thus leading to a high accuracy.
However, several drawbacks make it unsuitable to deal with big data.
Most important, this technique requires recording all the data for the leverage of data is calculated based on the individual difference compared to all the other data.
As a result, all the data are involved in calculating the leverages, which would cost much computation time when dealing with big data.

A solution to this problem is to draw uniform and random samples from the data set, calculating the ``leverage-based'' probabilities to re-weight samples in the same way as the biased sampling probabilities, and use the samples and the leverage-based probabilities to generate the final answer.
The expectation of the average of the samples, calculated by accumulating the products of the leverage-based probability and the sample value, is an unbiased estimate of the accurate average of the samples~\cite{Ma2013A}.
Since the distribution of the samples is considered to be the same as the whole data distribution~\cite{Neyman1937Outline}, 
the accurate average of the samples is considered the same as the accurate average of the data, suggesting that the average, calculated with the leverage-based probabilities and the samples, is an unbiased estimate of the accurate average of the data.

However, when dealing with big data, calculating the sample leverages requires recording all the samples, which would decrease efficiency.
A solution is to define the sample leverages according to the current and previous samples while sampling. We can set several variables to record the ``general conditions'' ($e.g.$, average or median) of the previous samples instead of all the samples to calculate the leverage for the current sample, which requires much less storage space.
However, this approach is sensitive to the sampling sequence, and samples with the same value may have different leverages.
For example, suppose the individual difference of a sample is defined as dividing the sample value by the sum of the values of the current sample and all the previous samples (at this time, only the average of the previous samples needs to be recorded). 
Considering a sampling sequence $\{10, 10, 1, 1, 1\}$, the leverages of 10 can be 1 or 0.5, while for the sampling sequence $\{1,1,1,10,10\}$, the leverages of 10 may be $\frac{10}{13}$ or $\frac{10}{23}$.
Certain samples of different sampling sequences may produce different leverages and different aggregation answers, leading to a poor robustness.

Another solution is to calculate the leverages off-line to accelerate the online processing. For example, similar to~\cite{AQUA1999SIG}~\cite{Sidirourgos2011SciBORQ}, we could refer to previous query results or compute summary synopses in advance.
However, the off-line processing may also be impractical, since it is usually too expensive when dealing with big data due to the constrained time and resources. Additionally, they may be less flexible when dealing with queries on new data sets.

Some other drawbacks also make the previous approaches less efficient when dealing with big data.
The degree of the leverage effects, $i.e.$, how much influence the leverages have on the aggregation answers, is fixed in~\cite{Ma2013A}. However, to obtain better results, the actual conditions of the data should be considered to determine whether the leverage effects should be ``strong'' or ``weak''. When a ``weak'' leverage effect is enough, applying a ``strong'' leverage effect leads to poor answers, and it is the same the other way round. Thus, the fixed degree of the leverage effects in~\cite{Ma2013A} would bring about a loss of accuracy to some degree.
Besides, leverages are calculated in a single way in~\cite{Ma2013A}, while data with different features should be assigned with different leverages due to their different contributions to the aggregation answers.
Moreover, to reflect the individual differences in the data, the biased sampling is adopted in~\cite{Ma2013A}, which is much more difficult to implement than uniform sampling, thus may decrease the efficiency when dealing with big data.

\noindent\textbf{Contributions.}
In this paper, we propose a novel leverage-based approximate aggregation approach to overcome the stated limitations, which efficiently computes aggregation answers with a precision assurance. To overcome the limitation and inherit the advantages of uniform sampling, we draw uniform samples, and use leverages to generate probabilities to re-weight the samples to reflect their individual differences.
To overcome the limitation of the traditional simple leverages, we divide the data into regions according to their features and contributions, and assign different leverages to them.
To increase the accuracy, we introduce an iteration scheme of improving two constrained estimators, which intelligently determines the degree of the leverage effects according to the actual conditions.
An objective function is constructed, which makes our approach insensitive to the sampling sequences and unnecessary to record samples.
Our main contributions in this paper are summarized by:

\begin{enumerate}
\vspace{-0.2em}
\item A novel methodology for a high-precision estimate is proposed, which involves generating two estimators using different methods to iteratively process constrained modulations according to the actual conditions of the data.


\item
A sophisticated leverage strategy which considers the nature of data is proposed, in which the data are divided into regions and appropriately handled.

\item An objective function is constructed with the leverages and the samples, which avoids the sensitivity of the sampling sequences as well as storing the samples.

\item We conducted experiments compared with a uniform sampling, and the results show that our approach achieves high-quality answers with only 1/3 of the same sample size.

\item
To the best of our knowledge, iterative leveraging is applied to data management for the first time.
\vspace{-0.2em}
\end{enumerate}

In this paper, we focused on \textsf{AVG} aggregation, as \textsf{AVG} aggregation is one of the most common aggregation operations. Meanwhile, the answer of \textsf{SUM} aggregation can be easily obtained by multiplying the average and the number of data points, which could be easily obtained from meta data or computed according to the data size. Other aggregation functions, such as extreme value aggregation, will be studied in detail in the future.
\\

\vspace{-0.8em}\noindent\textbf{Organization.} We overview our approach in Section~\ref{sec: overview}, and introduce the preprocessing calculation in Section~\ref{sec: preEstimation}. We introduce a sophisticated leverage strategy in Section~\ref{sec: biasofSamples}, and propose different modulation strategies for the iteration scheme according to the actual conditions of the data in Section~\ref{sec: modulationStrategy}. The core algorithm is proposed in Section~\ref{sec: LBIAAlgorithm}, and extensions of our approach are discussed in Section~\ref{sec: extensions}. We present the experimental results in Section~\ref{sec: experiment}, survey related work in Section~\ref{sec: related}, and finally conclude the whole paper in Section~\ref{sec: conclusion}.
\vspace{-0.2em}
\section{Overview}\label{sec: overview}
In this paper, we propose a novel methodology of obtaining high-precision estimates. Based on the methodology, we developed a system to process leverage-based \textsf{AVG} aggregation queries.
\vspace{-0.2em}
\subsection{Methodology}\label{sec: basicidea}
We generated two estimators using different estimation methods and evaluated the bias of the estimators, \emph{i.e.}, relations between the accurate value and the estimators. We then modulate these estimators towards the accurate value according to the bias conditions to obtain the proper answers.

In this chapter, we use normal distributions, the most common distribution which is consistent with most of the actual conditions~\cite{Lyon2014Why}, to illustrate the methodology. We suppose data are randomly distributed, and we draw uniform and random samples from the data set. In the ideal condition, the sample distribution is an unbiased estimation of the data distribution~\cite{Neyman1937Outline}. 
In this condition, the sampled data are also normally distributed in high probability, 
and we could use the symmetry of normal distributions to evaluate the deviations of the estimators according to the actual conditions of the data.
Although the accurate value is unavailable, from the distribution of the sampled data, we could tell whether an estimator is larger or smaller than the accurate answer, and which estimator is closer to the accurate answer (discussed in detail in Section~\ref{sec: modulationStrategy}).
Based on such relations, these two estimators are iteratively ``modulated'' towards the accurate value. The one with more deviation is modulated more in each iteration. When the estimators are approximately equal to each other, they arrive at the accurate answer, and the high-precision answer is obtained.

\begin{figure}
\center
  \includegraphics[width = 9cm]{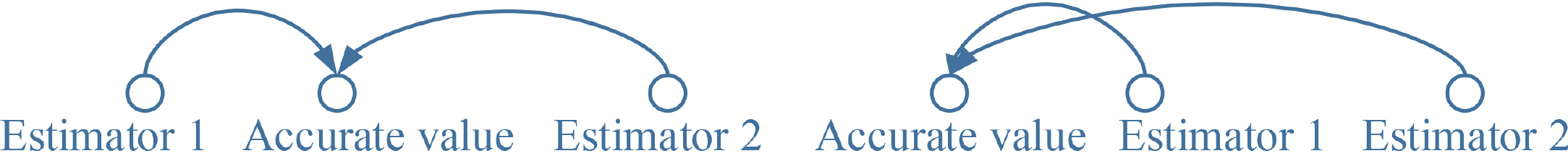}
  \vspace{-0.5em}
  \caption{Modulations of two conditions}
  \label{fig: modulation}
  \vspace{-1.5em}
\end{figure}

There are two cases about the relations between the accurate value and estimators, as shown in Fig.~\ref{fig: modulation}. One is that the accurate value is between the two estimators. The other is that the two estimators are on the same side of the accurate value.
In the first case, the larger of the two estimators is decreased, and the smaller one is increased. In the second case, the estimators are modulated in the same direction.
According to the deviations of the estimators, we tell which estimator is farther from the accurate value, then set different step lengths to obtain a high-precision, unbiased estimate. 

\begin{theorem}\label{th: methodology}
Consider two estimators $est_1$ and $est_2$ with deviations of $\varepsilon$ and $\varepsilon$+$\varepsilon'$ ($\varepsilon$, $\varepsilon'$$>$0) from the accurate value.
If the modulation step lengths of $est_1$ and $est_2$ are $\lambda s$ and $s$ (0$<$$\lambda$$<$1, $s$$>$0), respectively, an unbiased answer can be obtained when $\lambda$=$\varepsilon/(\varepsilon+\varepsilon')$.
\end{theorem}
\begin{proof}
We consider the first case in Fig.~\ref{fig: modulation}, and denote the accurate value as $acc$. Thus $est_1$=$acc$-$\varepsilon$, and $est_2$=$acc$+$\varepsilon$+$\varepsilon'$.
Suppose after $t$ rounds of modulation, $est_1$=$est_2$. Thus $est_1$=$acc$-$\varepsilon$+$\lambda ts$, and  $est_2$=$acc$+$\varepsilon$+$\varepsilon'$$ts$.
When $est_1$ and $est_2$ are both modulated to $acc$, an unbiased estimate is obtained, where $\lambda ts$-$\varepsilon$=0 and $\varepsilon$+$\varepsilon'$-$ts$=0, leading that $\varepsilon/(\varepsilon+\varepsilon')$=$\lambda$.
The proof of the second case is similar.
\end{proof}

\subsection{Introduction of the Leverage}
We introduce leverages to reflect the individual differences of the data, as well as their different contributions to the global answers. Here we use an example to illustrate the benefit.

\noindent\textbf{Example 1. }Consider a data set $\{1,$ 2, 3, 4, 5, 6, 7, 8, 9, 20\} and a sample set $\{2,$ 4, 6, 8, 20\} randomly generated from the data set. The accurate average of all the data is 6.5.
We process \textsf{AVG} aggregation on the sample set. The traditional uniform answer, and the leverage-based answer are computed as follows.\\

\vspace{-0.8em}
\noindent
\underline{\emph{(1) Traditional.}} The uniform answer is generated by equally dividing the sum, and we get an answer of 8.
In the data set, 20 is much larger when compared with other data, which can be defined as ``outliers''. Compared with other ``normal'' data (counted for 9/10), this kind of data count for only 1/10, which is less likely to be sampled than the ``normal'' data. However, once it is sampled, due to the extremely large value, it brings significant influence to the answer, leading to a deviation of the result.

\noindent
\underline{\emph{(2) Leverage-based.}} Considering the individual differences of 20 and other data, in the computation process, we use leverages to re-weighted the samples. In the sample set, each sample counts for 1/5. We regard 20 as the ``outlier'', and other samples as the ``normal sample''. To weaken the influence of the outlier to the global answer, we apply a small leverage, for example, 0.6, to 20, and re-weight sample 20 by 0.6*0.2=0.12. To ensure the sum of the probabilities to be 1, we compute the probability of each normal sample as (1-0.12)/4=0.22, with a leverage of 0.22/0.2=1.1. We then use the leveraged probabilities to compute the average as 0.12*20+0.22*(2+4+6+8)=6.8, which is much closer to the accurate average of 6.5. Due to the leverages, the influence of the sampled outlier is decreased, leading to an increase in the quality of the approximate answer.

\vspace{-0.2em}
\subsection{System Architecture}
According to the methodology above, we adopt two estimators, the sketch estimator ($sketch$), and the leverage-based estimator ($l$-$estimator$).
The sketch estimator, initially generated with a relaxed precision requirement, describes a ``rough picture'' of the aggregation answer.
The leverage-based estimator is calculated with samples and leverages, where the individual differences of samples are considered.

We establish a system to process \textsf{AVG} aggregation using $sketch$ and \emph{l-estimator}. Queries are of this form:
\textsf{SELECT AVG(column) FROM database WHERE desired precision},
where \textsf{desired precision} is the precision requirement indicated by the users.
The flow chart of our system is shown in Fig.~\ref{fig:system}.


When faced with big data, a centralized storage is impractical.
Thus, without any loss of generality, we propose the data to be stored in multiple machines, \emph{i.e.}, blocks. In this condition, to process the aggregation, it is effective to compute on each block, then gather the partial results to generate the final answer.
Considering this, we divide the main functions into three modules, Pre-estimation, Calculation, and Summarization.
The Pre-estimation module calculates the parameters for later computation. The Calculation module processes iterations to obtain partial answers on each block. The Summarization module collects the partial answers to generate the final aggregation answers. We now overview these modules.\\

\vspace{-0.8em}
\noindent\textbf{Pre-estimation module.}
This module calculates the sampling rate and the sketch estimator for later computations.
To satisfy the desired precision indicated by the users, we calculate a sampling rate to draw samples in the blocks. The sketch estimator is then generated with a relaxed precision as an overall picture of the final answer, which is to be later modulated to increase accuracy in blocks in the Calculation module.
Details of this module will be discussed in Section~\ref{sec: preEstimation}.


\noindent\textbf{Calculation module.}
The Calculation module mainly processes core computations on the blocks. A data division criteria (data boundaries) is established according to the distribution feature to divide data into different regions, thus samples with different features can be treated differently. In each block, samples are drawn according to the sampling rate. Based on the data boundaries and the sketch estimator, partial answers are iteratively computed on the blocks.

In each block, once the samples are picked, they fall into specific regions according to the data boundaries.
Only samples in certain regions, which are featured enough to represent the whole distribution, are further considered, since an approximate distribution could be determined from these samples.
Using these samples, the leverage-based estimator is generated to reflect the individual differences in the samples, and the \emph{l-estimator} and $sketch$ are iteratively modulated to generate high-precision answers.

\begin{figure}
\center
\includegraphics[width=8cm]{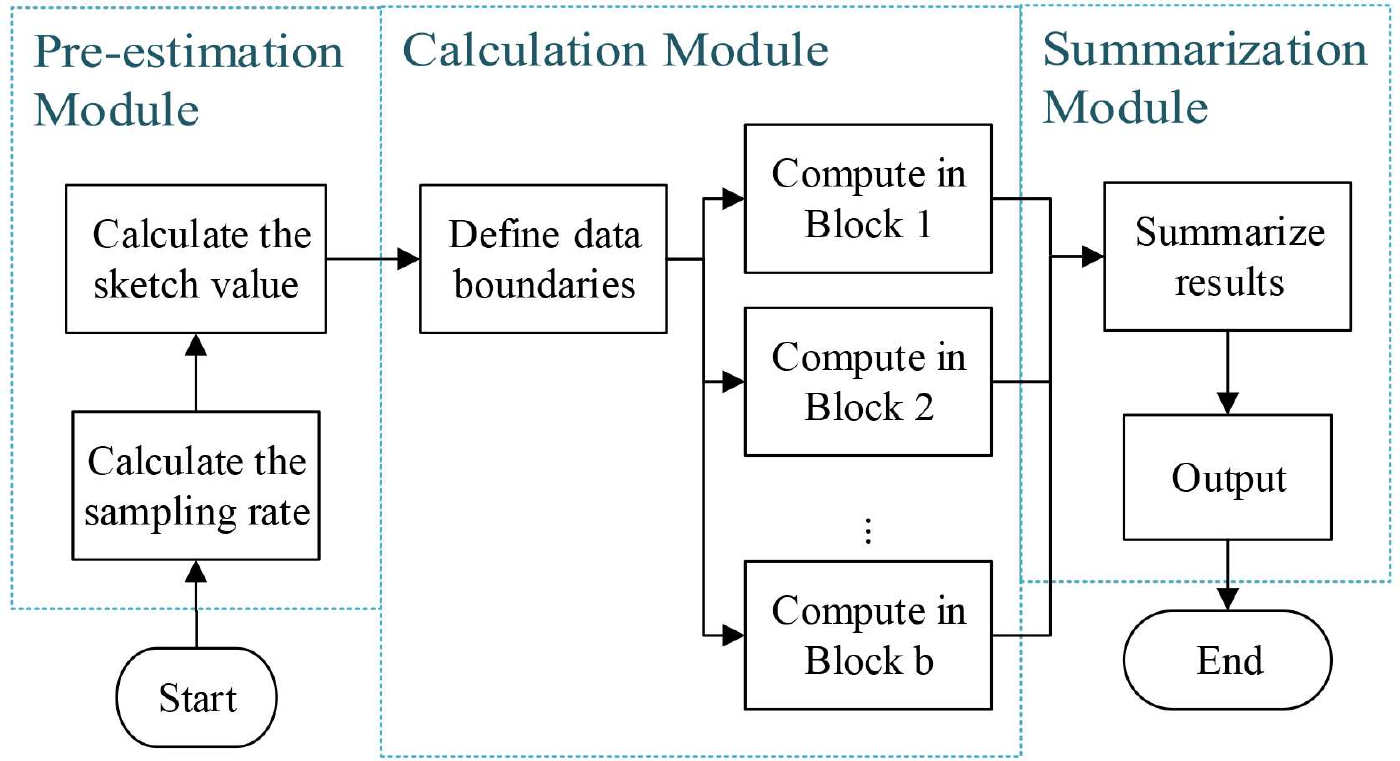}
\caption{System flow chart}
\label{fig:system}
\vspace{-1.5em}
\end{figure}

We discuss the measures for individual differences of the samples in Section~\ref{sec: biasofSamples}, where data boundaries and leverages are explained in detail.  We then illustrate the modulation strategies for the \emph{l-estimator} and $sketch$ in Section~\ref{sec: modulationStrategy}, and finally talk about the core algorithms to compute the proper partial answers in the blocks in Section~\ref{sec: LBIAAlgorithm}.

\noindent\textbf{Summarization module.}
This module collects the partial answers to generate the final answer. We denote the block set and the number of blocks as $B$ and $b$, respectively, and denote the partial answers of block $B_1$, $\cdots$, $B_b$ as $avg_1$, $\cdots$, $avg_b$, respectively.
Since these partial answers represent the average conditions in the blocks, and blocks with more data contribute more to the aggregation, for block $B_j$, the probability of $avg_j$ is set just positively relevant to the block size $|B_j|$. The final answer is thus calculated
as $\sum_{j=1}^{b}avg_j |B_j|/M$, where $M$ is the data size.\\

\vspace{-0.8em}
The main notations in this paper are summarized in Table~\ref{tb: notation}. In this paper, we assume that the blocks provide unbiased samples for their local data. For the ease of illustration, we assume data in the blocks are independent and identically distributed (\emph{i.i.d.}), and extend our approach to the non-\emph{i.i.d.} distributions in Section~\ref{sec: extension-noniid}.

\begin{table}
\vspace{-0.5em}
\centering
\caption{A sumary of the main notations}
\vspace{-0.5em}
\begin{tabular}{ll}
\hline\noalign{\smallskip}
  Symbol & Meaning\\
  \hline
  $e$ & Required precision, indicated by the users.\\
  $r$ & Sampling rate.\\
  $\mu$ & Accurate average aggregation answer. \\
  $\hat{\mu}$ & The value of \emph{l-estimator}.\\
  $\alpha$ & The leverage degree.\\
  $q$ & The leverage allocating parameter.\\
  $\lambda$ & The step length factor.\\
  $D$ & The objective function for iterations. \\
  $M$ & The data size. \\
  $X$ & The set of \textsf{S} samples. $X$=$\{x_i\}_{i=1}^u$. \\
  $Y$ & The set of \textsf{L} samples. $Y$=$\{y_j\}_{j=1}^v$. \\
  $fac$ & The normalization factor for \textsf{S} and \textsf{L} data.\\
  $thr$ & The iteration threshold. \\
  $dev$ & The deviation degree of $sketch_0$: $dev$=$|$\textsf{S}$|$/$|$\textsf{L}$|$.\\
  $p_1$, $p_2$ & Data boundary parameters. \\
  $|$\textsf{S}$|$, $|$\textsf{L}$|$ & The number of data in the \textsf{S} or \textsf{L} region. \\
  $sketch_0$ & The initial value of the sketch estimator.\\
  $\delta\alpha$, $\delta skech$ & The modulation step lengths for $\delta$ and $\alpha$.\\
\noalign{\smallskip}\hline
\end{tabular}\label{tb: notation}
\vspace{-1em}
\end{table}

We mainly discuss the normal distributions, since normal distributions are the most consistent with the actual situations~\cite{Lyon2014Why}. We also provide extensive discussions to show that our approach can also be adapted to other distributions in Section~\ref{sec: otherDistribution}, and experimentally evaluate the performance of our approach on some extreme conditions, such as uniform distribution and exponential distribution, in Section~\ref{sec: expotherDistribution}. Actually, many models assume that data are normally distributed, such as linear regression~\cite{Karlsson2007Introduction}, which assumes that the errors are normally distributed, and non-normal distributions can even be transferred to normal distributions~\cite{datatransformations}.

\vspace{-0.2em}
\section{Pre-estimation}\label{sec: preEstimation}
Pre-estimation module calculates the sampling rate and the sketch estimator for later computation.

\vspace{-0.5em}
\subsection{Sampling Rate}\label{sec: sampleRatioEstimation}
To satisfy the desired precision indicated by the users, the system calculates a sampling rate based on which the blocks draw samples.

For different aggregation tasks, the desired precision $e$ indicated by the users is different, and a proper sampling rate should be calculated accordingly.
We assume that the corresponding sample size is $m$. To calculate $m$, we introduce the confidence interval~\cite{Voit2001The}, which is a precision assurance to confirm that the accurate answer is in it.
\begin{definition}\label{def: confidenceinterval}
Define $\{z_1, z_2, \dots, z_{m}\}$ as a sample set generated from a normal distribution $N(\mu, \sigma^2)$, and $\bar{z}$ is the average of the samples. For confidence $\beta$, the confidence interval of $\mu$ is $(\bar{z}-u\frac{\sigma}{\sqrt{{m}}}, \bar{z}+u\frac{\sigma}{\sqrt{{m}}})$, where $\sigma$ is the standard deviation, and $u$ is a parameter determined by $\beta$.
\end{definition}

According to Neyman's principle \cite{Neyman1937Outline}, the confidence $\beta$ is specified in advance.
In our problem, the confidence interval is determined by the aggregation answer $\bar{z}$ and the desired precision $e$, since we would like the accurate answer in the interval of $(\bar{z}-e, \bar{z}+e)$, which is the confidence interval in our problem. According to Definition~\ref{def: confidenceinterval}, the length of the confidence interval is $2e$, where $e$=$u\frac{\sigma}{\sqrt{{m}}}$. The required sample size $m$ is then obtained: $m$=${u^2 \sigma^2}/{e^2}$, and the sampling rate $r$ is computed as
\vspace{-0.5em}
\begin{equation}\label{eq: sampleratio}
\vspace{-0.2em}
r = \frac{m}{M} = \frac{u^2\sigma^2}{Me^2},
\vspace{-0.2em}
\end{equation}
where $M$ is the number of data points, and $\sigma$ is the overall estimated standard deviation. We assume $M$ is known (actually, $M$ could be easily obtained from the meta data or computed according to the data size). To estimate $\sigma$, a small sample set is used, with a sample size indicated by the system in advance, and the samples are uniformly and randomly picked from each block with a sample size proportional to the block size.
Note that $\sigma$ is subject to error. However, in our approach, $\sigma$ only participates in estimating the sampling rate and establishing the data boundaries (Details are in Section~\ref{sec: leveragestrategy}), and does not evolve in aggregation computation. In this condition, it hardly has any effect on the answers, and we do not need to pay much attention to the accuracy guarantee of $\sigma$.





\vspace{-0.2em}
\subsection{Sketch Estimator}\label{sec: sketchValueEstimation}
The sketch estimator is generated with a pilot sample set as an overall picture of the final answer. It is used to determine the data boundaries and is modulated to increase the precision to obtain the proper aggregation answers in the blocks later in the Calculation module.

Denote the sketch estimator as $sketch$ and its initial value as $sketch_0$.
Note that an arbitrary sample size does not provide any definite precision
assurance. If $sketch_0$ is calculated with an arbitrary sample size, the later modulation of $sketch_0$ would bring uncertainty and precision loss to the final answers.
To ensure accuracy, $sketch_0$ is generated using a relaxed precision $t_e\cdot e$, where $t_e(t_e\textgreater 1)$, determined by the system, is the \emph{relaxed precision parameter}. In this condition, $sketch_0$ is provided with a relaxed confidence interval $(sketch_0-t_e$$e$, $sketch_0$+$t_e$$e)$. We generate $sketch_0$ in the same way as $\sigma$, where uniform samples are picked from each block with the sample size proportional to the block size. Similarly, the sampling rate for $sketch_0$ is obtained according to Eq.~(\ref{eq: sampleratio}). In this way, $sketch_0$ is obtained with a relaxed precision assurance. Such a $sketch_0$ is modulated in each block to increase the accuracy later in the Calculation module.\\

%
\vspace{-1em}

\section{Bias of Samples}\label{sec: biasofSamples}
We consider the bias and individual differences of the samples to increase the accuracy.
To save cost, samples are uniformly drawn. However, data act differently in aggregation, and regarding them uniformly brings a loss of accuracy. Thus, a re-weight processing is required.
In this section, we introduce how our approach reflects individual differences in the samples from the statistical perspective.
A sophisticated leverage strategy is introduced in Section~\ref{sec: leveragestrategy}, which considers the nature of data and divides the data into regions then handles them differently. 
We then illustrate how to use leverages to generate probabilities to reflect the individual differences in the samples in Section~\ref{sec: probabilityGeneration}.
\vspace{-0.2em}
\subsection{Leverage Strategy}\label{sec: leveragestrategy}
To overcome the limitations of the traditional leverages, inspired by \cite{error2014vldb}, we propose a sophisticated leverage strategy which considers the nature of data.
We divide the data into regions according to the data division criteria ($i.e.$, data boundaries), choose regions which are featured enough to represent the whole distribution, and assign various leverages to reflect individual differences in the samples.

\begin{figure}
\begin{minipage}[t]{0.5\linewidth}
\centering
\includegraphics[width=4.5cm]{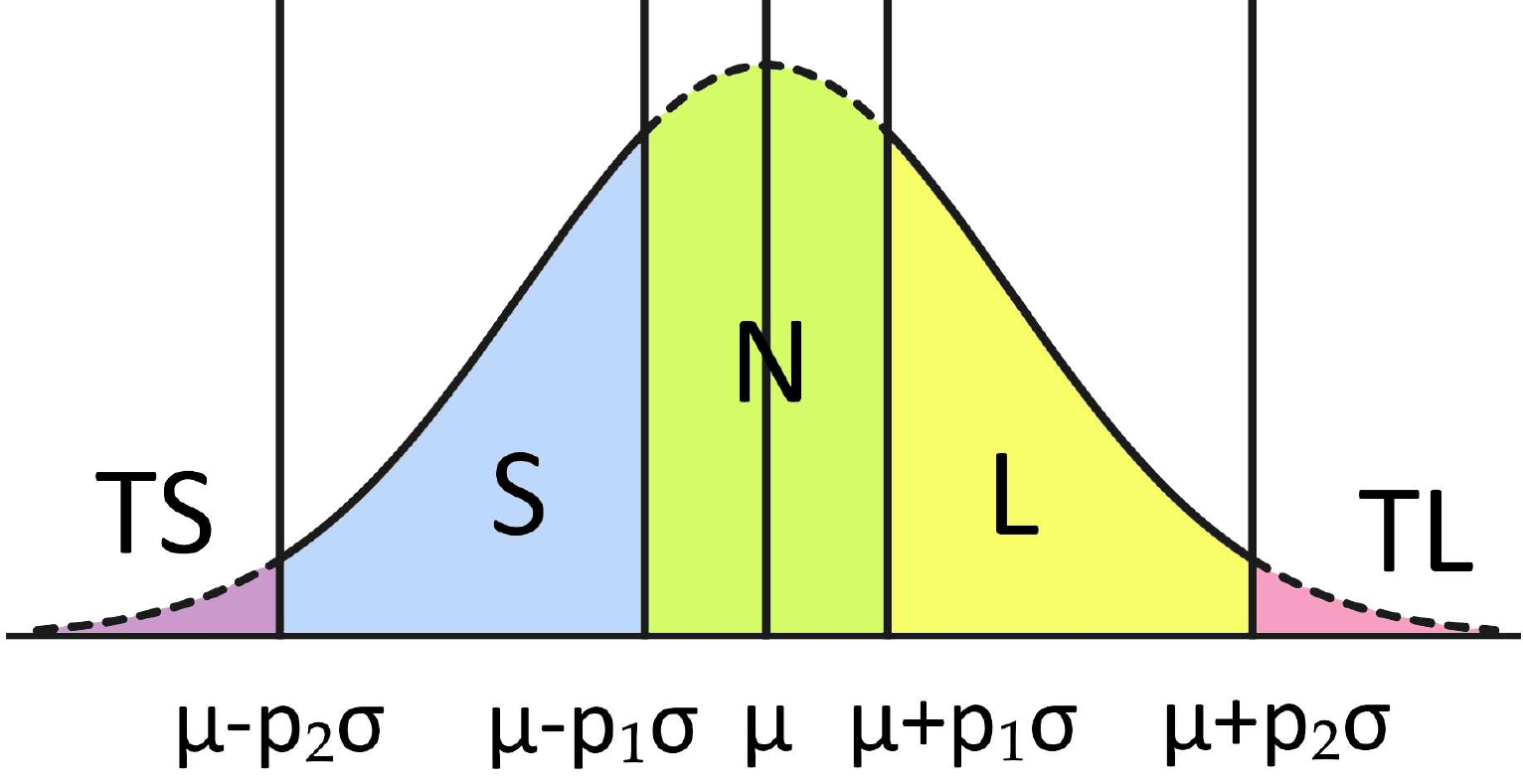}
  \vspace{-0.5em}
  \caption{Data division}\label{fig:datadivision}
  \vspace{-1.2em}
\end{minipage}%
\begin{minipage}[t]{0.5\linewidth}
\centering
\includegraphics[width=4.5cm]{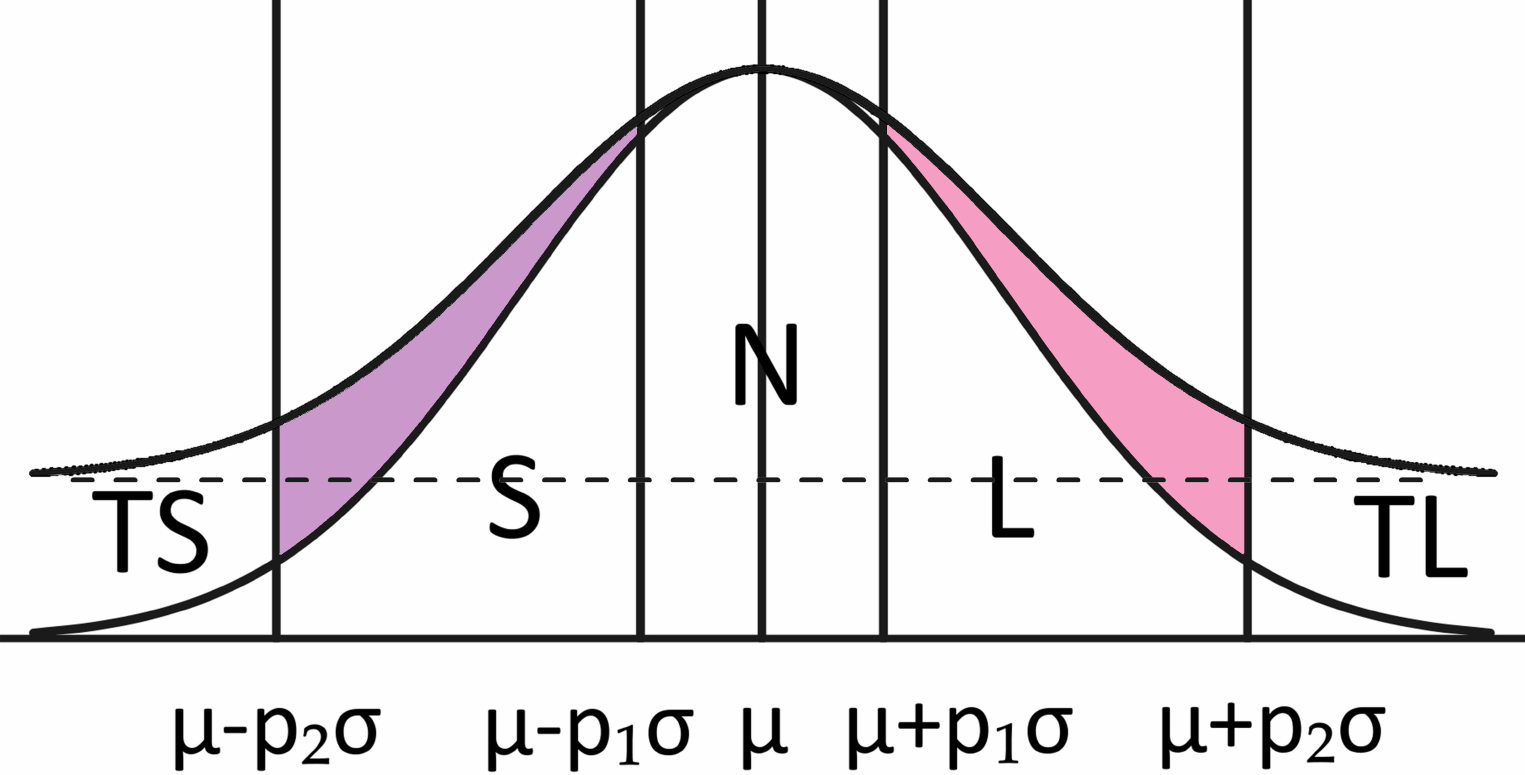}
  \vspace{-0.5em}
  \caption{Data contributions}\label{fig: distance}
  \vspace{-1.2em}
\end{minipage}
\vspace{-0.8em}
\end{figure}

%
%
%

\vspace{0.2em}
\subsubsection{Data Boundaries}\label{sec: dataBoundsDefinition}
To handle data with different features , we use data boundaries to distinguish the data.


Most of the existing approximate aggregation approaches handle samples identically, regardless of the differences among the
samples~\cite{online1997}~\cite{Wu2010Distributed}. 
However, data with different features contribute differently to the global answers in \textsf{AVG} and \textsf{SUM} aggregation, and neglecting the differences brings a loss of accuracy.
For example, in normal distributions, some data are large and can be easily picked, which significantly contributes to the global answers. Some data are even much larger than most of the other data,
but they count for limited proportions to be picked, which can be regarded as large outliers in \textsf{AVG} aggregation. However, once the large outliers are picked, due to their extremely large values, significant effects occur to the aggregation answers.


To treat data with different features, we consider the nature of data and divide the data into regions based on their values and positions in the normal distributions referring to the ``3$\sigma$ rule''~\cite{Friedrich1994The}. 
Since data out of  $(\mu-2\sigma, \mu+2\sigma)$ count for a limited proportion (about 4.6\%~\cite{Friedrich1994The}) and are too far away from the middle axis, which have limited contributions to the \textsf{AVG} aggregation, we regard them as outliers and do not consider the boundaries of $\mu\pm3\sigma$. Therefore, the boundaries of ``3$\sigma$'' divide the data distributions into 5 regions.
We use $sketch_0$ and the standard deviation $\sigma$ calculated in the Pre-estimation module to define the data boundaries.
To control the percentages of data in these regions, we set the data boundary parameters $p_1$ and $p_2$ ($0 \textless p_1 \textless p_2$) to adjust the data boundaries.
In this way, the proportion of data involved in the computation is controlled.

The data boundaries are shown in Fig.~\ref{fig:datadivision}, where the data are divided into the following 5 regions.




\noindent
\underline{\emph{(1) Too small}} (\textbf{TS}). Data in $(-\infty, sketch_0-p_2\sigma]$ are defined as ``too small data''.
Such data have extremely low values and can hardly be sampled due to their extremely low probabilities, thus can be treated as a kind of outlier in \textsf{AVG} aggregation, and their effects can be nearly neglected.

\noindent
\underline{\emph{(2) Small}} (\textbf{S}).
Data in $(sketch_0-p_2\sigma, sketch_0-p_1\sigma)$ are defined as ``small data''.
Such data count for a high proportion and have lower values than most of the others. 

\noindent
\underline{\emph{(3) Normal}} (\textbf{N}).
Data in $[sketch_0-p_1\sigma, sketch_0+p_1\sigma]$ are defined as ``normal data''. These data are symmetrical around the middle axis in the distribution and have higher probabilities to be sampled than data in other regions. 

\noindent
\underline{\emph{(4) Large}} (\textbf{L}).
Data in $(sketch_0 +p_1\sigma, sketch_0+p_2\sigma)$ are defined as ``large data''. Such data have higher values than most of the others
and count for a high proportion, thus significantly contribute to \textsf{AVG} aggregation.

\noindent
\underline{\emph{(5) Too large}} (\textbf{TL}).
Data in $[sketch_0+p_2\sigma,+\infty)$ are defined as ``too large data'', which have extremely high values but can hardly be sampled due to the extremely low probabilities. Thus in \textsf{AVG} aggregation they can also be regarded as a kind of outlier. However, different from the \textsf{TS} data, once such data are sampled, a significant influence might occur to the aggregation answers due to their extremely high values. Thus, in \textsf{AVG} aggregation, such a significant influence should be considered to be eliminated or properly handled.


\vspace{0.2em}
\subsubsection{Leverage Assignment}\label{sec: leverageAssignment}
Due to the different contributions of the data in \textsf{AVG} aggregation, we use different leverages to reflect the differences.
We use data in \textsf{S} and \textsf{L} to represent the distributions and directly discard the other data, since the \textsf{S} and \textsf{L} data contribute much to the \textsf{AVG} aggregation, and the shape of the distributions can even be approximately predicted from the \textsf{S} and \textsf{L} regions. As shown in Fig.~\ref{fig:datadivision}, \textsf{S} and \textsf{L} are symmetric in their distribution, approaching the middle axis from the left and right sides, respectively, and they both account for high proportions. Meanwhile, the parameters of the distributions ($\mu$ and $\sigma$) are included in the shapes of the \textsf{S} and \textsf{L} regions, and other regions can even be approximately speculated with \textsf{S} and \textsf{L}, as the dotted line in Fig.~\ref{fig:datadivision}.

In practice, the symmetry of \textsf{S} and \textsf{L} data in the normal distribution model can be used in estimating the data distribution, for normal distributions are the most consistent with the actual conditions~\cite{Lyon2014Why}. Even if an actual data is not normally distributed, the distribution is usually similar to normal distributions or can be generated by superimposing several normal distributions. Such actual data sets are usually symmetric around the average axis, and we could use such symmetry to process aggregation computing to increase both accuracy and efficiency. 

Due to the different contributions of the samples in different regions, we assign different leverages to the \textsf{S} and \textsf{L} data.
We assign values farther from the middle axis with greater leverages. The reason is that, although they have less probabilities, they contribute more to the shapes of the normal distributions when considering the formation of the normal distributions. As shown in Fig.~\ref{fig: distance}, values farther from the middle axis have more reflection on whether a normal distribution is ``short and fat'' or ``tall and thin'', and such information describes the distributions.
Data farther from the middle axis provide more information about the shapes of the normal distributions. Thus, larger leverages are assigned to them, and smaller leverages are assigned to the closer ones.

Considering a sample set $A$=$\{a_i\}_{i=1}^m$, for each $a_i$, we introduce its \emph{deviation factor} $h_i$ to calculate its leverage.
Commonly, score $h_i$ is used to define whether the data are outliers~\cite{Hoaglin1977The}~\cite{Velleman1981Efficient}~\cite{Weisberg1986Influential}. Inspired by \cite{Ma2013A}, we use it to calculate the leverages. For sample $a_i$, $h_{i}$=${a_i^2}/{\sum_{j=1}^m a_j^2}$. Obviously, for positive values~\footnote{For the ease of discussion, we assume that all the data are positive. For aggregation with negative data, we translate the distribution along the x axis by the distance of $d$ to make all the data positive to process the computation, and then move back the answer by the distance of $d$ to generate the final answer.
}, $h_{i}$ is positively correlating to the values. 
For the \textsf{S} and \textsf{L} data, we assign larger leverages to the samples farther from the middle axis. Considering this, for data $a_i$, if it is \textsf{S}, its leverage score is $1-h_{i}$; if it is \textsf{L}, the leverage score is $h_{i}$.
%

Theoretically, the estimator cannot achieve the same accuracy guarantee when compared with uniform sampling, for only the \textsf{S} and \textsf{L} samples are participated in computation.
However, the precision loss is not that significant, since the chosen samples can effectively represent the whole distribution.
Meanwhile, proper leverages are assigned to the \textsf{S} and \textsf{L} samples to reflect the individual differences, and the accuracy is even increased.
Furthermore, a much better accuracy is achieved in practice.
Compared with uniform sampling, our approach achieves high-quality answers using only $\frac{1}{3}$ sample size, and actually, only the \textsf{S} and \textsf{L} samples in the 1/3 samples are participated in computation. Details are in Section~\ref{sec: experiment4Lev}.

In our approach, to save computation time, only the \textsf{S} and \textsf{L} data involve calculation; to increase the precision, samples are assigned with different leverages based on their different contributions in the aggregation.
\vspace{0.2em}
\subsubsection{Leverage Normalization}\label{sec: scale}
Even though we assign different leverages to the \textsf{S} and \textsf{L} data to reflect their individual differences, since such leverages do not satisfy some constraints inherited from the probability calculation, we could not directly use the leverage scores in the probability generation, and a normalization for the leverages is required based on these constraints.  

The following theorem describes \textbf{Constraint 1}, which is inherited from the probability generation.
\begin{theorem} \label{th: levsum}
The sum of leverages equals to 1.
\end{theorem}
\noindent\begin{proof}
For each $a_i$ in the sample set $A$=$\{a_i\}_{i=1}^{m}$, its probability is of this form (details are discussed in Section~\ref{sec: probabilityGeneration}): $prob_i$=$\alpha lev_i$+$(1$-$\alpha)/m$ $(0$$<$$\alpha$$<$$1)$, where $lev_i$ is the leverage, and $1/m$ is the uniform sampling probability.
When accumulating the probabilities of all the samples, $\sum prob_i$=$1$, that is, $\sum$$prob_i$=$\alpha$$\sum$$lev_i$+$\sum$$(1$-$\alpha)/m$=$\alpha\sum$$lev_i$+$1$-$\alpha$=$1$, leading that $\sum$$lev_i$=$1$.
\end{proof}

However, according to THEOREM~\ref{th: levsum} we could not obtain the concrete leverage sums of the \textsf{S} and \textsf{L} data, for their ratio is not obtained. Thus, we propose \textbf{Constraint 2}.

\noindent\textbf{Constraint 2: }The leverage sum of the samples in a specified region is proportional to the number of samples in it.

We establish this constraint according to the following consideration.
Data boundaries are established using $sketch_0$, the initial value of the sketch estimator. The deviation of $sketch_0$ leads to a difference in $|$\textsf{S}$|$ and $|$\textsf{L}$|$. From Fig.~\ref{fig: bias}, we observe that the accurate average value $\mu$ is closer to the region with more data.
Thus, a larger sum of leverages is desired to the region with more data, and we set it proportional to the number of samples in that region. As a result, the leverage-based estimator will be closer to the region with more data, thus it will be much closer to the accurate answer.

According to these discussions, the leverage normalization is performed as follows:
\begin{itemize}
\item \underline{Step 1:} Leverage sum calculation. Get the sum of the leverage scores of the \textsf{S} and \textsf{L} data.
\item \underline{Step 2:} Theoretical sum calculation. Calculate the theoretical sum of leverages for the \textsf{S} and \textsf{L} data based on the two constraints we proposed.
\item \underline{Step 3:} Normalization factor calculation. Divide the sum of the leverage scores by the theoretical sum of the leverages to calculate the \emph{normalization factors} $fac$ for \textsf{S} and \textsf{L}.
\item \underline{Step 4:} Leverage normalization. For each \textsf{S} and \textsf{L} sample, divide its leverage by the corresponding normalization factor $fac$.
\end{itemize}

With the normalized leverages, the probabilities are generated to reflect the individual differences of the samples.
\vspace{0.2em}
\subsubsection{Sensitivity of $sketch_0$}\label{sec: biasDegreeEvaluation}%
Based on the previous discussions, we note that $sketch_0$ is important for the aggregation answers, for the data boundaries are established with $sketch_0$, which directly influences the classification of samples.
A bad $sketch_0$ may lead to a large difference between $|\textsf{S}|$ and $|\textsf{L}|$ and a large difference between the allocated sums of the leverages for the \textsf{S} and \textsf{L} samples. In this condition, the leverage effects of the region with more samples are too strong, leading to an over-modulation of the leverage effects over the aggregation answers.

A severe deviation of $sketch_0$ may happen due to unbalanced sampling.
Meanwhile, a pilot sample set is drawn to calculate $sketch_0$ in the Pre-estimation module, where uniform samples may have a significant influence on $sketch_0$. In this condition, if an outlier is picked, a significant deviation may happen to $sketch_0$.

To overcome the limitation of the sensitivity of $sketch_0$, we introduce the \emph{deviation degree}, denoted as $dev$, to evaluate the deviation of $sketch_0$, and introduce the \emph{leverage allocating} parameter $q$ to balance such a deviation by controlling the allocated sum of the leverages of the \textsf{S} and \textsf{L} samples.
We calculate $dev$ as $dev$ = $\frac{|\textsf{S}|}{|\textsf{L}|}$, and $dev$ within a system-specified range, \emph{e.g.}, $(0.99, 1.01)$,  indicates approximately no much deviation of $sketch_0$, while $dev$ out of the range denotes the deviation.

When an obvious deviation exists, \emph{e.g.}, $dev$ $\textless$ $0.94$, or $dev$ $\textgreater$ $1.06$, the leverage effect is too strong.
To weaken it, we use the leverage allocating parameter $q$ to control the allocated sum of the leverages (denoted by $levSum$) of \textsf{S} and \textsf{L} in the leverage normalization: $\frac{levSum_S}{levSum_L} = q\frac{u}{v}$.
We determine $q$ according to the actual conditions.
Generally, $q$ is set to 1. When an obvious deviation of $sketch_0$ occurs, we use a positive value $q'$ to generate $q$.
If $|\textsf{S}|$ $\textgreater$ $|\textsf{L}|$, we decrease the allocated sum of the leverages to the \textsf{S} data, and $q$ = $\frac{1}{q'}$; otherwise, $q$ = $q'$.
A large $q'$ is required, since the leverage tuning is subtle. Meanwhile, $q'$ should not be too large, since a too large $q'$ leads to a large sum of leverages allocated to the region with less data, which decreases the accuracy.
Actually, due to the confidence assurance of $sketch_0$, the difference between $|\textsf{\textsf{S}}|$ and $|\textsf{\textsf{L}}|$ is limited, leading to a not too large $q'$. Considering these, we vary $q'$ in $[5,10]$ in practice according to the deviation conditions of $sketch_0$.
In this way, we shrink the leverage effects of the region with more data to balance the too strong leverage effects.
Based on such a mechanism, our approach can detect and reduce the obvious deviation of $sketch_0$. Thus, the limitation of the sensitivity of $sketch_0$ is overcome.

An effective leverage strategy should understand the nature of data, divide the data into regions, and handle them differently.
In this paper, we propose a sophisticated leverage strategy. Various leverages are assigned to samples to reflect their individual differences. As a result, high-quality answers can be obtained with a small sample size.
\vspace{-0.2em}
\subsection{Probabilities Generation}\label{sec: probabilityGeneration}
In our approach, samples are uniformly picked. Inspired by the SLEV algorithm~\cite{Ma2013A}, we use leverages to re-weight samples to reflect their individual differences. In this subsection, we discuss how to generate re-weighted probabilities with the normalized leverages and the uniform sampling probabilities.

For sample $a_i$ in the sample set $A$ = $\{a_i\}_{i=1}^m$, let $lev_i$ denote the leverage of $a_i$, and let $unif_i$ denote the uniform distribution sampling probability (\emph{i.e.}, $unif_i$ = $1/m$, for all $i\in[m]$).
The re-weighted probability of $a_i$ is of the form
\vspace{-0.5em}
\begin{equation}\label{eq: probability}
prob_i = \alpha lev_i + (1 - \alpha)unif_i, \alpha \in (0,1),
\vspace{-0.5em}
\end{equation}
where $\alpha$ is the \emph{leverage degree}, which indicates the intensity of the leverage effect. 
The aggregation answer is then obtained by accumulating the product of probabilities and values, $\sum_{i=1}^m$ $prob_i\cdot a_i$.

\begin{table}
\vspace{-0.5em}
\center
\caption{Intermediate results of $l$-$estimator$ in Example 1}
\begin{tabular}{llllll}
\hline\noalign{\smallskip}
\textsf{Region}&\textsf{Val}&\textsf{OriLev}&\textsf{Fac}&\textsf{NorLev}&\textsf{Prob}\\ \hline\noalign{\smallskip}
\multirow{2}{*}{\textsf{S}}
&4
& 89/105
& \multirow{2}{*}{$169/70$}
& ${178}/{507}$
& $\frac{178}{507}\alpha+\frac{1-\alpha}{3}$\\
&5 & ${16}/{21}$  &
& ${160}/{507}$
& $\frac{160}{507}\alpha+\frac{1-\alpha}{3}$\\
\textsf{L}&8 & ${64}/{105}$ & $64/35$ & ${1}/{3}$ & $\frac{1}{3}\alpha+\frac{1-\alpha}{3}$\\ \hline
\end{tabular}\label{tb: example}
\vspace{-1em}
\end{table}

Here we illustrate the leverage effects and the process of leverage-based aggregation by the following example.\\

\vspace{-0.8em}
\noindent\textbf{Example 1. }Consider a data set $\{1,$ 2, 2, 3, 4, 4, 5, 5, 6, 6, 7, 8, 9, 10, 15\} and a sample set $\{2,$ 3, 4, 5, 6, 7, 8, 15\} randomly generated from the data set. The accurate average of all the data is 5.8.
We now generate the traditional uniform, and the leverage-based aggregation answers.\\

\vspace{-0.8em}
\noindent
\underline{\emph{(1) Traditional.}} The aggregation answer is generated by equally dividing the sum, and we get an answer of 6.25. The value 15 participates in the computation, which produces a deviation of the result due to its extremely large value.

\noindent
\underline{\emph{(2) Leverage-based.}} Suppose $sketch_0$ is 6.2, $p_1\sigma$=1, and $p_2\sigma$=3. Thus, the range of the \textsf{S} data is $(3.2, 5.2)$, and the range of the \textsf{L} data is $(7.2, 9.2)$.
According to our principled leverage approach, we know that only 4, 5, and 8 participate in the computation, where 4 and 5 fall in the \textsf{S} region, and 8 is in the \textsf{L} region.
The calculation processes are recorded in Table~\ref{tb: example}.
To generate the leverage-based probabilities of the samples, we first calculate the original leverages (\textsf{OriLev} in Table~\ref{tb: example}), then calculate the normalization factors (\textsf{Fac}) for the \textsf{S} and \textsf{L} data.
After that, we obtain the normalized leverages (\textsf{NorLev}). Finally, the probabilities (\textsf{Prob}) are generated with leverages, $\alpha$, and the uniform sampling probabilities.
Suppose $\alpha$ = $0.1$. By accumulating the products of the values and probabilities, we obtain the aggregation answer of 5.67. Due to the leverage effects, this answer is much closer to the accurate average of 5.8.

In this paper, we introduce leverages to reflect the individual differences in the samples to overcome the limitation of the uniform sampling. The leverage effect is controlled through the modulation of the leverage degree $\alpha$, which is crucial to the quality of the answers.
Using a fixed $\alpha$ means no modulation ability over the leverage effects, and a bad $\alpha$ leads to a low accuracy.
For example, if the aggregation answer calculated with uniform sampling probabilities are very close to $\mu$, only slight leverage effects are required over the aggregation answer, since only a little modulation is required. At this time, a large $\alpha$ produces an inaccuracy to the answer.

\begin{figure}
  \center
  \includegraphics[width=8cm]{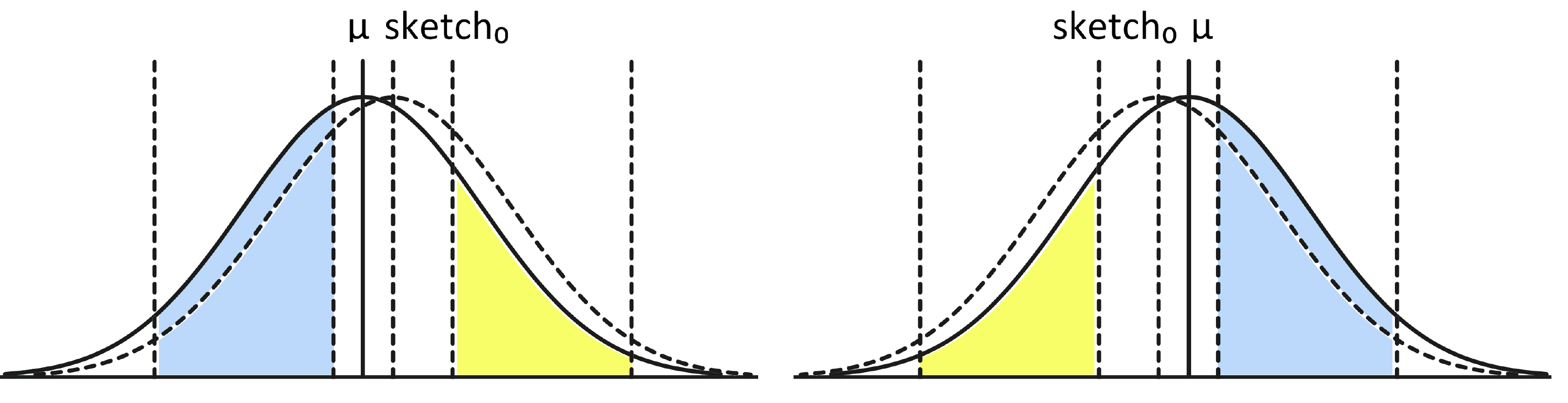}
  \vspace{-0.6em}
  \caption{Deviation of $sketch_0$. Solid lines: real distributions; dotted lines: estimated distributions;
  Shadows: number of data in \textsf{S} and \textsf{L}.}\label{fig: bias}
  \vspace{-1em}
\end{figure}

The quality of $\alpha$ has great influence over the aggregation answers, and the difficulty lies in that the actual conditions of the samples should be considered when deciding $\alpha$.
\vspace{-0.2em}
\section{Modulations}\label{sec: modulationStrategy}
A proper $\alpha$ is crucial to the quality of the aggregation answers, thus an intelligent mechanism is in great demand to determine $\alpha$. In our approach, modulations are processed according to the actual conditions of the samples to compute a good $\alpha$.



As already discussed, $\alpha$ relies on the actual conditions and can be hardly directly computed.
To obtain a proper $\alpha$ to achieve a high-quality leverage-based aggregation answer, we adopt the methodology proposed in Section~\ref{sec: basicidea}.
We adopt the leverage-based estimator $l$-$estimator$ and the sketch estimator $sketch$, and modulate them in the directions of $\mu$ to gradually increase the precision.  
As discussed in Section~\ref{sec: basicidea}, the deviations of the estimators are evaluated, and the estimator with more deviation from $\mu$ is modulated more in each iteration. When they are approximately equal to each other, they are both approximately arrived at the accurate value $\mu$, and a proper answer, as well as a good $\alpha$, is obtained.
To process the iterative modulations, we construct an objective function with leverages and samples, which avoids calculating the leverages while sampling and does not require recording samples, leading to our approach insensitive to the sampling sequences.

In this section, we discuss the objective function, illustrate how to evaluate the deviations of the \emph{l-estimator} and $sketch$, generate the modulation strategies according to the actual conditions of data, and finally, discuss how to determine the modulation step lengths.

\vspace{-0.2em}
\subsection{Function Construction}
We construct an objective function through a subtraction of the \emph{l-estimator} and \emph{sketch}. 
According to Section~\ref{sec: basicidea}, the optimization goal is the function value approaching 0 with the estimators evolving, modulated towards $\mu$.

We denote the value of the \emph{l-estimator} by $\hat{\mu}$, and set the initial value of $sketch$ to $sketch_0$, which is calculated in the Pre-estimation module.
We generate $\hat{\mu}$ with samples, leverages, $\alpha$, and the uniform sampling probabilities, and generate $sketch$ by modulating $sketch_0$.
To evaluate the deviation between these two estimators, we construct an objective function by subtracting $sketch$ from $\hat{\mu}$.
From the discussions in Section~\ref{sec: leveragestrategy}, only the \textsf{S} and \textsf{L} data are involved in the computing, while other data are directly discarded. For the \textsf{S} and \textsf{L} samples in an aggregation work, suppose $|\textsf{S}|$ = $u$ and $|\textsf{L}|$ = $v$. We use the \textsf{S} and \textsf{L} samples to generate the leverage-based answer $\hat{\mu}$, and the following theorem holds.

\begin{theorem}[The Leverage-based Answer]\label{th: function}
\vspace{-0.2em}
Denote the set of the \textsf{S} samples as $X$=$\{x_i\}_{i=1}^{u}$,
and the set of the \textsf{L} samples as $Y$=$\{y_j\}_{j=1}^v$. The leverage-based answer $\hat{\mu}$ is computed with a function of $\alpha$:
\vspace{-0.5em}
$$\hat{\mu}=f(\alpha)=k \alpha + c, $$
where $k$=$(\frac
 {(\sum x_i^2 + \sum y_j^2)\sum x_i-\sum x_i^3}
 {(1+\frac{v}{qu})(u(\sum x_i^2 + \sum y_j^2)-\sum x_i^2)}+\frac{v\sum y_j^3}{(qu+v)\sum y_j^2})
 -\frac{1}{u+v}(\sum x_i + \sum y_j)$, and $c$=$\frac{1}{u+v}(\sum x_i + \sum y_j)$.
\end{theorem}

We denote the difference between $\hat{\mu}$ and $sketch$ by $D$. According to THEOREM~\ref{th: function}, there exists 
\vspace{-0.2em}
\begin{equation}\label{eq: functionD}
D=\hat{\mu}-sketch = f(\alpha)-sketch=k \alpha+c-sketch
\vspace{-0.2em}
\end{equation}


Initially, $\alpha$ is set to 0. From THEOREM~\ref{th: function} we know that the initial value of $\hat{\mu}$ is $c$, which also stands for the aggregation answer calculated with the \textsf{S}, \textsf{L} samples and the uniform sampling probabilities, without leverages.

Note that the parameters in $D$ ($i.e.$, $k$ and $c$) are computed with $u$ and $v$, as well as the sum, square sum, cube sum of the \textsf{S} and \textsf{L} samples, and these variables can be computed while sampling. It indicates that the storage space for samples is totally unnecessary. Meanwhile, leverages are not directly calculated while sampling, leading to our approach insensitive to the sampling sequences.

According to $D$, the $l$-$estimator$ and $sketch$ are iteratively modulated approaching the directions of $\mu$, respectively, and the precision gradually increases.
\vspace{-0.2em}
\subsection{Deviation Evaluation}\label{sec: deviationevaluation}
We now discuss how to evaluate the deviations of $sketch$ and $c$ (the initial values of $\hat{\mu}$). We obtain two indicators for further processing. One is the relation among $sketch_0$, $c$, and $\mu$, which reveals the modulation directions of $\alpha$ and $sketch$. The other is the estimators' deviation conditions from $\mu$, \emph{i.e.}, which estimator is farther from $\mu$. Based on these indicators, the modulations are processed on \emph{l-estimator} and $sketch$.

Due to the symmetry of normal distributions, we could evaluate the deviation between $sketch_0$ and $\mu$ from the relation of $|$\textsf{S}$|$ and $|$\textsf{L}$|$, where in ideal conditions $|$\textsf{S}$|$=$|$\textsf{L}$|$. We can also evaluate the deviation between $c$ and $sketch_0$ through the initialization of the objective function $D$, and finally infer the relation of $sketch_0$, $c$, and $\mu$ according to the results of the former two steps.

\noindent
\underline{\emph{(1) The relation between $sketch_0$ and $\mu$.}}
The relation of $|$\textsf{S}$|$ and $|$\textsf{L}$|$ is evaluated to obtain the relation between $sketch_0$ and $\mu$, for the deviation of $sketch_0$ leads to a difference between the numbers of the data in \textsf{S} and \textsf{L}, as shown in Fig.~\ref{fig: bias}.


The \textsf{S} and \textsf{L} regions are defined by the data boundaries generated with $sketch_0$. Under ideal conditions, when $sketch_0$ is accurate, $|$\textsf{S}$|$=$|$\textsf{L}$|$, due to the symmetry of the \textsf{S} and \textsf{L} data in the distribution.
However, in practice, $sketch_0$ has a deviation from $\mu$, leading to $|\textsf{S}|$$\neq$$|\textsf{L}|$. Thus, according to the relation between $|\textsf{S}|$ and $|\textsf{L}|$, we could evaluate the deviation between $sketch_0$ and $\mu$.
$|$\textsf{S}$|\textgreater|$\textsf{L}$|$ indicates $sketch_0$$>$$\mu$ (as shown in the left of Fig.~\ref{fig: bias}), where $sketch_0$ should be increased; $|$\textsf{S}$|\textless |$\textsf{L}$|$ indicates $sketch_0$$<$$\mu$ (as shown in the right of Fig.~\ref{fig: bias}), where $sketch_0$ should be decreased.

\noindent
\underline{\emph{(2)The relation between $c$ and $sketch_0$.}}
To determine the modulation direction of $\alpha$, we evaluate the difference between $sketch_0$ and $c$ by initializing $D$. The initial value of $\hat{\mu}$ is $c$, and $\hat{\mu}$ is modulated through $\alpha$.
We denote the initial function value of $D$ by $D^0$. According to Eq.~(\ref{eq: functionD}), $D^0$=$c-sketch_0$, which reveals the relation between $c$ and $sketch_0$. $D^0$$>$0 indicates $c$$>$$sketch_0$; otherwise, $c$$<$$sketch_0$.


\noindent\underline{\emph{(3)The relation between $sketch_0$, $c$, and $\mu$.}}
We now obtain the relation of $sketch_0$ and $\mu$, and the relation of $sketch_0$ and $c$. We combine them together to obtain the relation of
$sketch_0$, $c$, and $\mu$, which reveals the modulation directions of $\alpha$ and $sketch$, as well as the relation of their modulation step lengths.
\vspace{-0.2em}
\subsection{Modulation Strategies}
\label{sec: modulationStrategies_Condisitons}
We previously discuss the deviation evaluation of the estimators and obtained the relations between $sketch_0$, $c$, and $\mu$. We now illustrate how to use such relations to generate different modulation strategies according to the different conditions of the samples.

The modulation strategies include the modulation directions for $\alpha$ and $sketch$ (to increase, or decrease), and the relations of the modulation step lengths (for $\hat{\mu}$ and $sketch$, which one is modulated more in each  iteration).
Suppose the modulation step lengths of $\alpha$ and $sketch$ are $\delta \alpha$ and $\delta sketch$, respectively, which indicates how much to change in a
round of modulation.
For the ease of discussion, the step lengths are set positive. When $\alpha$ or $sketch$ is to be increased, its step length is added to it; otherwise, subtracted from it.

Different conditions of the samples lead to different relations between $sketch_0$, $c$, and $\mu$, and different modulation strategies are required for $sketch$ and $\alpha$.
Note that although we suppose each block provides unbiased samples for its local data, in practice, unbalanced sampling happens by small probabilities. To fit this scenario, we considered the unbalanced sampling and generate the corresponding modulation strategies.
We obtain the relations of $sketch_0$, $c$, and $\mu$ according to the method proposed in the last subsection, and then derive the relation of $\delta \alpha$ and $\delta sketch$ based on the optimal goal of the iteration ($D$$\rightarrow$0). We now discuss different cases and the corresponding modulation strategies as follows.
%
%
%

\noindent
\textbf{Case 1:} $D^0\textless0$, $|\textsf{S}|\textless|\textsf{L}|$: $c\textless sketch_0 \textless \mu$; $k\delta\alpha\textgreater\delta sketch$.

Since both $c$ and $sketch_0$ are smaller than $\mu$, they both should be increased, thus
$D$=$k(0$+$\delta\alpha$)+$c$-$(sketch_0$+$\delta sketch)$$\rightarrow$0, leading to $k \delta \alpha \textgreater \delta sketch$.
In this case, unbalanced sampling happens. Since $sketch_0\textless\mu$ and $|\textsf{L}|\textgreater|\textsf{S}|$, as shown in the right of Fig.~\ref{fig: bias}, $c$ should be on the right of $sketch_0$. However, contradiction exists since $c\textless sketch_0$, indicating an unbalanced sampling. This seldom happens. Both $sketch$ and $c$ increase, and $c$ increases more to balance the bias.

\noindent
\textbf{Case 2:} $D^0 \textless 0$, $|\textsf{S}|$$\textgreater$$|\textsf{L}|$: $c, \mu \textless sketch_0$; $k\delta\alpha$+$\delta sketch \textgreater 0$.

We increase $c$ and decrease $sketch_0$. Thus, $D$=$k(\delta\alpha$+$0)$+$c$-$(sketch_0$-$\delta sketch)$$\rightarrow$$0$,
leading that $k\delta \alpha$+$\delta sketch \textgreater 0$. When $k$$>$0, such a relation always holds; otherwise, $\delta sketch$$\textgreater$$|k\delta \alpha|$. In this case, the relation of $c$ and $\mu$ is unknown, and the modulation direction of $\hat{\mu}$ could not be directly determined. However, uniform sampling probabilities cannot reflect the individual differences in the samples,
which poorly works when compared with the leverage-based probabilities. Meanwhile, unbalanced sampling does not occur, and we do not need a negative $\alpha$ to balance the sampling bias. Therefore, we slightly increase $\hat{\mu}$ for better answers.

\noindent
\textbf{Case 3:} $D^0 \textgreater 0$, $|\textsf{S}|\textless|\textsf{L}|$: $c, \mu \textgreater sketch_0$; $k\delta \alpha \textless \delta sketch$.

In this case both $c$ and $sketch_0$ are increased. Thus, $D = k\delta\alpha$+$c-(sketch_0$+$\delta sketch)\rightarrow 0$, leading to $k\delta \alpha \textless \delta sketch$. The explanations are similar to Case 2.

\noindent
\textbf{Case 4:} $D^0 \textgreater 0$, $|\textsf{S}| \textgreater |\textsf{L}|$: $c \textgreater sketch_0 \textgreater \mu$; $k\delta \alpha \textgreater \delta sketch$.

Both $c$ and $sketch_0$ should be decreased, thus, $D=k(0 - \delta \alpha) + c - (sketch_0 - \delta sketch)\rightarrow 0$,
leading that $k\delta\alpha\textgreater\delta sketch$.
Similar to Case 1, unbalanced sampling also occurs. When we decrease $sketch$, we decrease $k\delta\alpha$ more, and $\alpha$
is negative to balance such unbalanced sampling.

\noindent
\textbf{Case 5:} $|\textsf{S}|$$\approx$$|\textsf{L}|$: return $sketch_0$.

In this case, \textsf{S} and \textsf{L} are approximately balanced, indicating that $sketch_0$ works well as the data division criteria for it is very
close to $\mu$. We do not use any further process, just return $sketch_0$ as the aggregation answer.

In our approach, different modulation strategies are generated according to the actual conditions, where both $\hat{\mu}$ and $sketch$ are modulated approaching to $\mu$ to increase the precision.
\vspace{-0.2em}
\subsection{Step Lengths}\label{sec: steplength}
Based on the modulation strategies of the different conditions above, $\delta\alpha$ and $\delta sketch$ can be determined.
Since long step lengths may lead to missing proper answers, while short step lengths result in a slow convergence, analogous to the gradient descent method \cite{Burges2005Learning}, we developed a self-tuning mechanism for the step lengths to ensure both the accuracy and the convergence speed.

\begin{algorithm}[h]
    \caption{Phase 1: Sampling}
    \small
    \begin{algorithmic}[1]
    \REQUIRE~~\\
        $j$: the block id; $r$: the sampling rate; \\
        $param_{db}$: data boundary information;\\
    \ENSURE~~\\
        $j$: the block id;\\
        $param_\textsf{S}$: $\{$counter, sum, squareSum, cubeSum$\}$;\\
        $param_\textsf{L}$: $\{$counter, sum, squareSum, cubeSum$\}$;\\

    \STATE Initialize $param_S$, $param_L$;\\
    \STATE $m\leftarrow r |B_j|$;   // Calculate the sample size.\\
    \FOR{$i\leftarrow1$ to $m$}
        \STATE Draw a sample $a$;\\
        \STATE Classify $a$;    //$a$ is classified according to $param_{db}$\\
        \IF{$a$ belongs to \textsf{S}}
            \STATE updateParams($a$, $param_S$);\\
        \ENDIF
        \IF{$a$ belongs to \textsf{L}}
            \STATE updateParams($a$, $param_L$);\\
        \ENDIF
        \STATE Drop $a$;\\
    \ENDFOR
    \end{algorithmic}\label{algorithm: lbia-phase1}

    updateParams($a$, $param$)
    \begin{algorithmic}[1]
    \STATE $param$.counter $\leftarrow$ $param$.counter+1;
    \STATE $param$.sum $\leftarrow$ $param$.sum+$a$;
    \STATE $param$.squareSum $\leftarrow$ $param$.squareSum+$a^2$;
    \STATE $param$.cubeSum $\leftarrow$ $param$.cubeSum+$a^3$;

    \end{algorithmic}

\end{algorithm}

We determine the step lengths according to $D$ and set a \emph{convergence speed} $\eta$ $(\eta$$\in$$(0,1))$, where $D$ reduces to $\eta D$ after an iteration. In this paper, we set $\eta$ to 0.5, which means $D$ reduces to half after each iteration.
According to the optimal goal of $D$, we generate the relation among $sketch$, $\delta{sketch}$, $\alpha$, $\delta{\alpha}$, and $D$.
To ensure the relation between $|k\alpha|$ and $\delta sketch$ generated above, we introduce a $\lambda$ (0$<$$\lambda$$<$1) as the \emph{step length factor}. The smaller one of $|k \alpha|$ and $\delta sketch$ is set to the larger one multiplied by $\lambda$.
In this way, we determine the step lengths in the current iteration using $sketch$ and $\alpha$ in the last iteration.

For example, initially, $D$=$c$-$sketch_0$.
In the first iteration, there exists $k(0$$\pm$$\delta{\alpha})$+$c$-$(sketch_0$$\pm$$\delta{sketch})$=$\eta D$, and $min\{|k \alpha|, \delta sketch\}$ = $\lambda$$\cdot$$max\{|k\alpha|$, $\delta sketch\}$. We obtain $\delta{\alpha}$ and $\delta{sketch}$ with these two equations and update $sketch$, $\alpha$, and $D$.
Similarly, the second iteration is then processed with the new $sketch$, $\alpha$, and $D$, etc.

\noindent\textbf{Determination of $\lambda$.}
The deviations of the \emph{l-estimator} and $sketch$ are evaluated to determine $\lambda$. As discussed above, a severe deviation of $sketch_0$ leads to a large difference between $|\textsf{S}|$ and $|\textsf{L}|$ as well as the strong leverage effects on 
\emph{l-estimator}.
According to THEOREM~\ref{th: methodology}, when determining $\lambda$, we should consider the severe deviation of $sketch_0$ and the strong leverage effects, and adopt a different $\lambda$ based on the actual conditions.
However, we introduce the \emph{leverage allocating} parameter $q$ to shrink the severe deviation of $sketch_0$ and the leverage effects in Section~\ref{sec: biasDegreeEvaluation}. Since $q$ is determined according to the actual conditions of the samples, a fixed $\lambda$ is sufficient. 

%
%

\noindent
\textbf{Terminal Condition.}
$D$ is reduced to half after each iteration, hence $D$ is steadily approaching 0 at a high rate of convergence.
We introduce a threshold $thr (thr\textgreater 0)$ for $D$. When $|D|$ is no greater than $thr$, the iteration halts.
\section{Core Algorithm}\label{sec: LBIAAlgorithm}

\vspace{-0.2em}
We introduce the core algorithm of our approach based on the leverage mechanism and the iterative modulation scheme. The algorithm runs in each computing block to compute the partial answer with an evolving $\alpha$ and finally returns a proper aggregation answer of the current block.
Two phases are included, \emph{i.e.}, the sampling phase, and the iteration phase.

\vspace{-0.5em}

\subsection{Phase 1: Sampling}
In the sampling phase, samples are drawn, and then the regions they falling in are determined. Two arrays, $param_\textsf{S}$ and $param_\textsf{L}$, are set to record information of the \textsf{S} and \textsf{L} samples, including the counter, sum, square sum, and cube sum.
Once a sample is picked, if they fall in the \textsf{S} or \textsf{L} region, the corresponding array is updated; Otherwise, the sample is directly discarded, for it does not participate in the computation.
The pseudo code is shown in
Algorithm~\ref{algorithm: lbia-phase1}.

\begin{algorithm}[h]
    \caption{Phase 2: Iteration}
    \small
    \begin{algorithmic}[1]
    \REQUIRE~~\\
        $param_S$: $\{$counter, sum, squareSum, cubeSum$\}$;\\
        $param_L$: $\{$counter, sum, squareSum, cubeSum$\}$;\\
        $sketch_0$: the initial value of the sketch estimator;\\
        $thr$: the threshold parameter for iteration; \\
        $\eta$: the convergence rate;
    \ENSURE~~\\
        j: the block id; $avg$: the aggregation answer;\\

    \IF{($param_S$.counter $\approx$ $param_L$.counter)}
        \STATE $avg \leftarrow sketch_0$;\\
        \RETURN $(j, avg, param_S, param_L)$;
    \ENDIF
    \STATE Construct the objective function $D$;\\
    \STATE Determine modulation strategies;\\
    \STATE $\alpha\leftarrow 0$, $sketch\leftarrow sketch_0$, $d\leftarrow D^0$;\\
    \WHILE{$|d|\textgreater thr$}
        \STATE Calculate $\delta{sketch}$ and $\delta\alpha$;\\        \STATE $d\leftarrow \eta d$, $sketch\leftarrow sketch+\delta{sketch}$,
        $\alpha\leftarrow\alpha+\delta{\alpha}$;\\
    \ENDWHILE
    \STATE $avg\leftarrow k \alpha+c$;\\
    \end{algorithmic}\label{algorithm: lbia-phase2}
\end{algorithm}

Two arrays, $param_\textsf{S}$ and $param_\textsf{L}$, are initialized to record the information for the \textsf{S} and \textsf{L} samples (Line 1). The required sample size in this block is then computed (Line 2). Next, samples are drawn and classified (Line 3-10).
Once a sample $a$ is drawn, it is classified according to the data boundaries (Line 4, 5). If it is \textsf{S} or \textsf{L} data, the corresponding parameters ($param_\textsf{S}$ or $param_\textsf{L}$) are updated, where the algorithm adds 1, $a$, $a^2$, $a^3$ to the counter, sum, square sum, and cube sum, respectively (Line 6-9). The samples are then dropped (Line 10).

\noindent
\textbf{Complexity analysis.} 
According to the previous discussions, this phase requires $O(m)$ time, where $m$ is the sample size.

Instead of recording all the samples, the information of the samples are included in $param_\textsf{S}$ and $param_\textsf{L}$, which are used to compute $k$ and $c$ in the objective function $D$ according to THEOREM~\ref{th: function} later in the next phase.

\vspace{-0.2em}
\subsection{Phase 2: Iteration}
In the iteration phase, modulations are processed iteratively, and a proper aggregation answer is obtained. The pseudo code is shown in Algorithm~\ref{algorithm: lbia-phase2}.

Initially, whether $|\textsf{S}|$ is approximately equal to $|\textsf{L}|$ is evaluated; if $|\textsf{S}|\approx|\textsf{L}|$, $sketch_0$ is directly returned as a proper aggregation
answer of the current block (Line 1-3), for $sketch_0$ is very close to $\mu$; otherwise, the algorithm continues.
The function $D$ is constructed (Line 4), and the modulation strategies for $sketch$ and $\alpha$ are determined (Line 5).
After initialization (Line 6), it processes (Line 7-9): for each iteration, $D$ decreases by a speed of $\eta$, based on which the step lengths, $\delta sketch$ and $\delta \alpha$, are calculated for the current iteration (Line 8); then the parameters are updated for the next iteration (Line 9). When the function value of $D$ is below the threshold $thr$, a good $\alpha$ is obtained, and the aggregation answer of the current block is obtained with this $\alpha$ (Line 10).

\noindent
\textbf{Upper bound for iteration.} As discussed in Section~\ref{sec: steplength}, in each iteration, $D$ is decreased by the speed of $\eta$. When $|D|$ is no more than the threshold $thr$, the iteration halts.
We suppose the iteration time is $t$. There exist $(\frac{1}{2})^t|D^0|$$\leq$$thr$ and $(\frac{1}{2})^{t-1}|D^0|$$\textgreater$$thr$. Thus, $t$=$\lceil
log(\frac{|D^0|}{thr})\rceil$.

\noindent\textbf{Convergency.} As discussed in Section~\ref{sec: steplength}, the modulation step lengths $\delta\alpha$ and $\delta sketch$ are calculated based on the modulation objective (D$\rightarrow$0) and the relation of $c$, $sketch_0$, and $\mu$. Meanwhile, the difference between $\hat{\mu}$ and $sketch$ decreases by a convergence speed of $\eta$ ($\eta$$\in$$(0,1)$) in each iteration, indicating that $D$ is converged to 0.

\noindent
\textbf{Complexity analysis.}
According to previous discussions, the iteration phase is processed in $O(\lceil log(\frac{|D^0|}{thr})\rceil)$.

We use $param_\textsf{S}$ and $param_\textsf{L}$ to construct the objective function $D$ for the iterations, which not only requires no storage space for the sampled data but also makes our approach insensitive to the sampling sequences. Due to the iteration scheme, $\alpha$ is intelligently determined according to the actual conditions, leading to a high accuracy and efficiency.

\section{Extensions}\label{sec: extensions}
\vspace{-0.2em}
Our approach can be extended to fit more scenarios.

\vspace{-0.2em}
\subsection{Extension to Online Aggregation.}
\label{sec: furtherComputation}
Our approach can be extended to the online mode to support further computation after accomplishing the current computation.
In each computing block, $param_S$ and $param_L$ are stored to record the counter, sum, square sum, cube sum of the \textsf{S} and \textsf{L} samples, respectively, instead of storing all the samples. Further computations can be processed based on $param_S$ and $param_L$.
After the current round of computations is accomplished, if users would like to continue computations to obtain an answer with a higher precision, then our system can continue computations based on the data boundaries, $param_S$, and $param_L$, for the information of the previous samples are recorded in $param_S$ and $param_L$.
For \textsf{S} and \textsf{L} samples in the new round of computations, similar updates are applied to the counter, sum, square sum, and cube sum in $param_S$ and $param_L$. Based on $param_S$ and $param_L$, the system processes the iterations to achieve a higher precision.

\vspace{-0.5em}
\subsection{Extension to Other Distributions}\label{sec: otherDistribution}
Our approach is proposed based on normal distributions, since actual data are usually subjected to: 1) normal distributions, 2) similar normal distributions, or 3) distributions generated by superimposing several normal distributions.

Our approach can also handle non-normal distributions due to the leverages, the iteration scheme, and the precision assurance of $sketch_0$.
In our approach, $sketch_0$ is generated as a ``rough picture'' of the aggregation value with a relaxed precision, which provides a constraint for the result. Due to the confidence assurance of $sketch_0$, the final answer could not be far away from it.
Meanwhile, leverages are assigned to samples to reflect their individual differences, which overcomes the flaws of uniform sampling and increases the precision of the $l$-$esitimator$. In iterations, $l$-$estimator$ and $sketch$ are gradually modulated to increase the accuracy.



In Section~\ref{sec: expotherDistribution} we test the validation of our approach on some other distributions, such as, the exponential distribution.
In our approach, we use $|\textsf{S}|$ and $|\textsf{L}|$ to evaluate the deviation of $sketch_0$. When handling some extreme distributions (although we hardly process \textsf{AVG} aggregation on these distributions in practice), such as, $f(x)$=$2^x(x$$\textgreater$0), $|\textsf{S}|$ and $|\textsf{L}|$ dramatically vary with the increase in $x$. In this condition, a loss of accuracy may occur due to the high increasing rate of $f(x)$.

To solve this problem, we can utilize the confidence interval of $sketch_0$ to generate a modulation boundary for the estimators. The confidence interval provides an assurance of $\mu$ in this interval. It also indicates that $\mu$ can hardly be out of the range. However, when a severe difference of $|\textsf{S}|$ and $|\textsf{L}|$ occurs, the computed aggregation answer will be out of the interval due to the strong leverage effects.
Note that such a feature can be used to test whether there is a high increasing (or decreasing) rate of $f(x)$. Moreover, we can evaluate how much the aggregation answer excesses the interval to evaluate the increasing (or decreasing) rate of $f(x)$, then choose different $q$, the leverage allocating parameter, to optimize the leverage effects.

The extreme condition evaluation and a more detailed definition of parameter $q$ will be studied in the future.

\vspace{-0.2em}
\subsection{Extension to \emph{Non-i.i.d.} Data}\label{sec: extension-noniid}
In this paper, we suppose that data in the  blocks are \emph{i.i.d.}. We now consider the local variance of blocks and propose ideas about the \textsf{AVG} aggregation on \emph{non-i.i.d.} blocks. Improvements are mainly from the following two aspects.

\noindent\textbf{Different sampling rates.}
For aggregation on \emph{non-i.i.d.} distributions, to balance the accuracy and efficiency, we consider the local variance of blocks and apply different sampling rates to them.
Inspired by~\cite{Haas2004A}, we apply the blocks, where data show much dispersion (or, variability), with a large sample size to obtain enough information to describe the data distributions. Considering that such dispersion is reflected by $\sigma$, we use $\sigma$ to compute leverages to reflect the local variance in the blocks, and blocks with higher $\sigma$ are applied with higher sampling rates.

For block $B_i$, we denote its leverage as $blev_i$, and its sampling rate is computed with $blev_i$, the overall sampling rate $r$, the data size $M$, and the data size $|B_i|$.
Similarly to the leverages in ~\cite{Ma2013A}, we set the leverage of $B_i$ proportional to $\sigma_i^2$. Since such leverages are directly used in computing the sampling rates, to avoid the sampling rates of 0, we set $blev_i$ as $\frac{1+\sigma^2}{b+\sum_{i=1}^{b}\sigma_i^2}$, and compute the sampling rate of $B_i$ as $rM$$\cdot$$blev_i/|B_i|$.
To calculated $\sigma_i$, in the Pre-estimation module, a small sample set is drawn randomly and uniformly from $B_i$. Meanwhile, the samples from the blocks are collected to generate the overall sampling rate $r$.

\noindent\textbf{Different data boundaries.}
Since the identical data boundaries work poorly for different distributions, for \emph{non-i.i.d.} distributions, in Pre-estimation module, a pilot sample set is drawn in each block to calculate $sketch_0$ and $\sigma$ to generate different data boundaries. Similar iterations are then processed to compute the proper answers in these blocks.
\vspace{-0.6em}
\subsection{Extension to Other Aggregation Functions}
\vspace{-0.2em}
In this paper, we focused on the \textsf{AVG} aggregation, and the \textsf{SUM} can be easily obtained by multiplying the average and the data size $M$.
Meanwhile, the leverage-based approach can also work for other aggregations, such as \textsf{MIN}, \textsf{MAX} and \textsf{GROUP BY}, where different leverages are applied to reflect the individual differences of data to increase the accuracy. We will continue our research on this line of research. Currently the work of extreme value aggregation, \textsf{MAX} and \textsf{MIN}, is in progress, and we now give a brief introduction.


We use a similar framework, and the main differences include 1) the recorded information (only the extreme value is recorded in each block), and 2) the sampling rate, where leverages are used to generate different sampling rates according to the local variance and the general conditions of blocks.

As discussed in Section~\ref{sec: extension-noniid}, the sampling rates are generated based on the local variance. Blocks which exhibit a higher $\sigma$ should be sampled more than blocks with a lower $\sigma$. Meanwhile, considering the particularity of extreme value aggregation, the general conditions of the blocks should also be considered, since data in some specific blocks may be higher or lower than other blocks in general.
We take \textsf{MAX} aggregation for example. The \textsf{MAX} value is more likely to be in the blocks with generally higher values, while it is less possible to be in the blocks with lower values.

Under this condition, only considering the local variance is insufficient when generating the sampling rates. Thus, we will develop a leverage-based sampling rate which considers the local variance and the general conditions of the blocks.
The local variance is reflected by $\sigma$, and the general conditions of the blocks can be described using the average or median, which indicates a general condition of the data in the block. For blocks with generally higher values, larger leverages are assigned to the sampling rate, while blocks with generally lower data are assigned with smaller leverages.

\vspace{-0.3em}
\subsection{Extension to Distributed Systems}
\vspace{-0.2em}
In some scenarios, big data are distributed on multiple machines, $e.g.$, HDFS.
Our approach can be easily extended for distributed aggregation due to the architectural features, Meanwhile, it also provides convenience to deal with big data, for there is no requirement to store the samples.

Distributed aggregation could be implemented by performing sample-based aggregation on each machine and then collecting the partial results.
We use an example to illustrate.
Considering a transnational corporation, massive data are stored distributedly in its subsidiaries all over the world, which brings the requirement of handling big data over its subsidiary corporations.
When processing aggregation, according to our approach, computations are processed in each subsidiary. The center node then collects the partial results to generate the final answer.
\vspace{-0.6em}
\subsection{Extension of Time Constraint}
In some applications, users set the time constraints for the computation, such as \cite{Sidirourgos2011SciBORQ}\cite{AgarwalBlinkDB}. Our system could accomplish aggregation with a time constraint with small adjustments.
According to the workload, the relationship of the sample size and the run time could be obtained, based on which our system calculates the required sample size within the time constraint. The system then generates the precision assurance--the confidence interval--to ensure accuracy.
\vspace{-0.2em}
\section{Experiment Evaluation}\label{sec: experiment}
We conducted extensive experiments to evaluate the performance of our approach (an iterative scheme for leverage-based approximate aggregation, \textsf{ISLA} for short). We first compared \textsf{ISLA} with the uniform sampling method to evaluate the effects of the leverages. Due to the leverages, our approach can achieve high quality answers with a small sample size.
We then tested the impact of the parameters on the performance of our approach.
Next, we compared \textsf{ISLA} with the \emph{measure-biased} technique proposed in \emph{sample+seek}~\cite{sampleseek2016}, the state-of-art approach.
Finally, we evaluated the performance of \textsf{ISLA} on other distributions as well as the real data.

We compared the approximate aggregation answers with the accurate answer to evaluate the quality of the estimate answer.
However, when dealing with big data, it is impractical to compute accurate answers. Therefore, we used synthetic data generated with a determined average $\mu$ as the golden truth.
We generated data in normal distribution $N(\mu, \sigma^2)$ with an accurate average of $\mu$. We then compare the estimated average with $\mu$ to evaluate whether our approach computes a high-quality aggregation answer. Without special illustration, we set $\mu$ to 100 and set $\sigma$ to 20.

\begin{figure}[H]
\makeatletter\def\@captype{figure}
\center
\subfigure[Varing precision]{
    \label{fig:precisionVary}
    \includegraphics[width = 4cm]{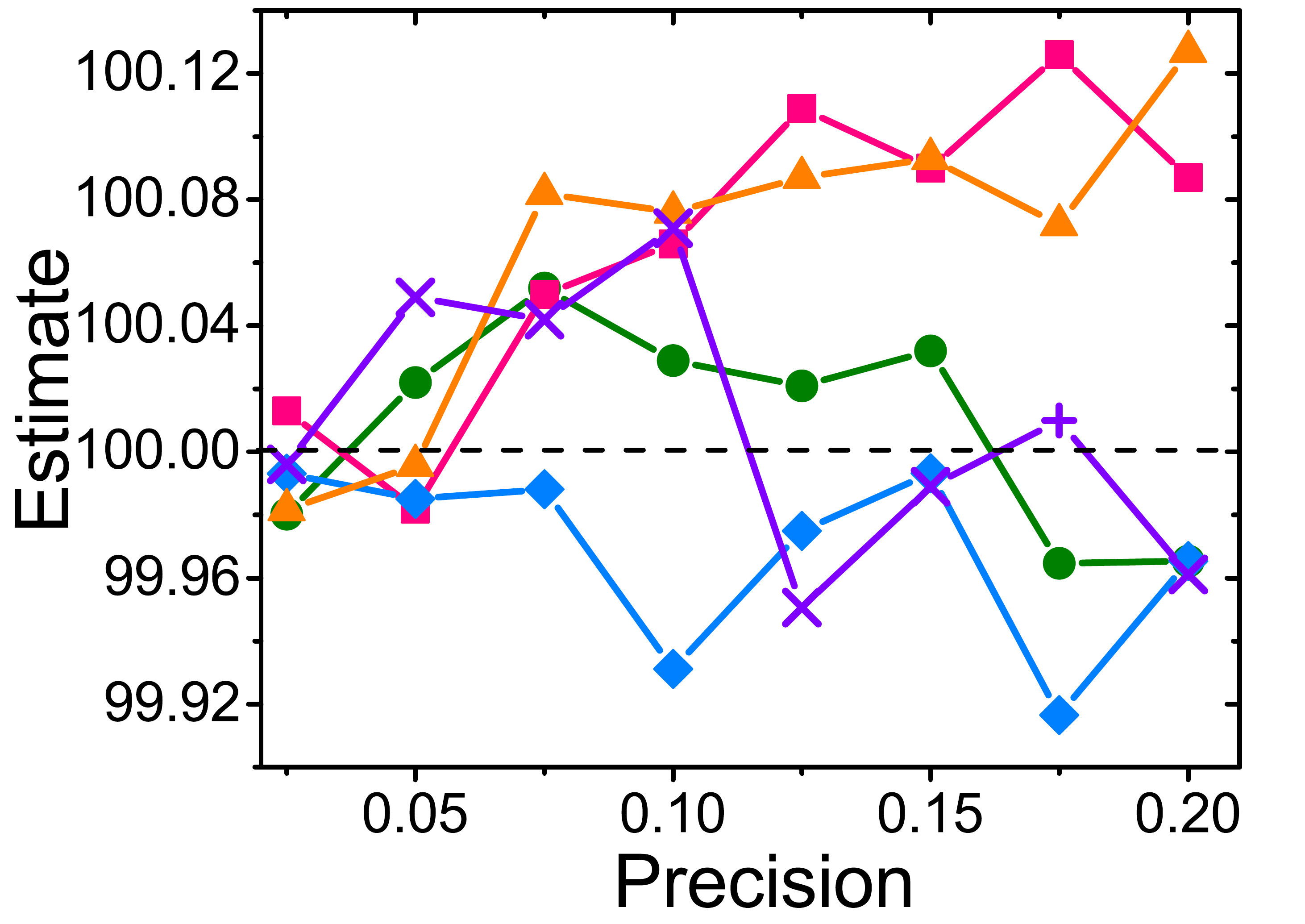}}
\subfigure[Varing confidence]{
    \label{fig:confidenceVary}
    \includegraphics[width = 4cm]{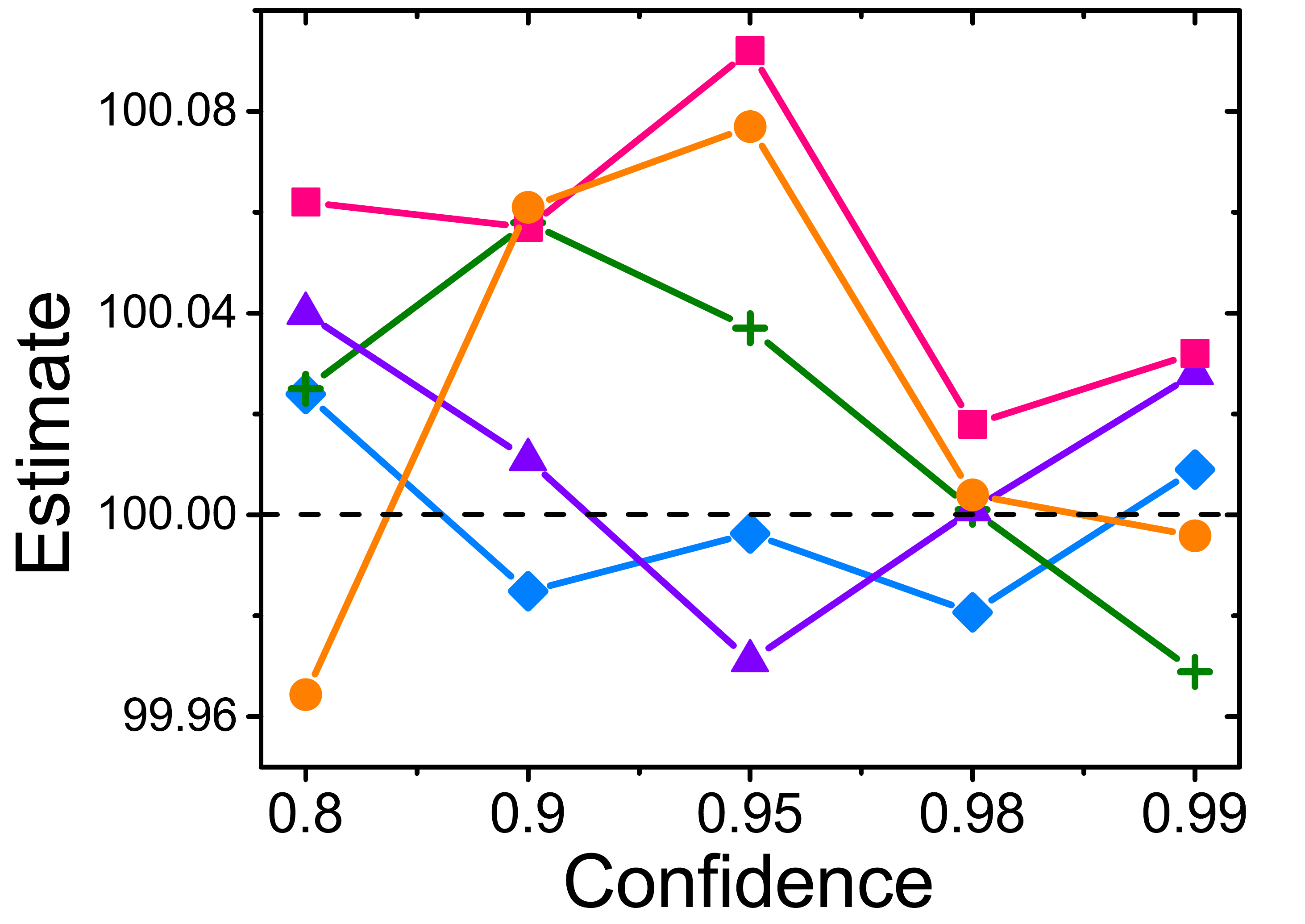}}
\subfigure[Varing number of blocks]{
    \label{fig:blockNumVary}
    \includegraphics[width = 4cm]{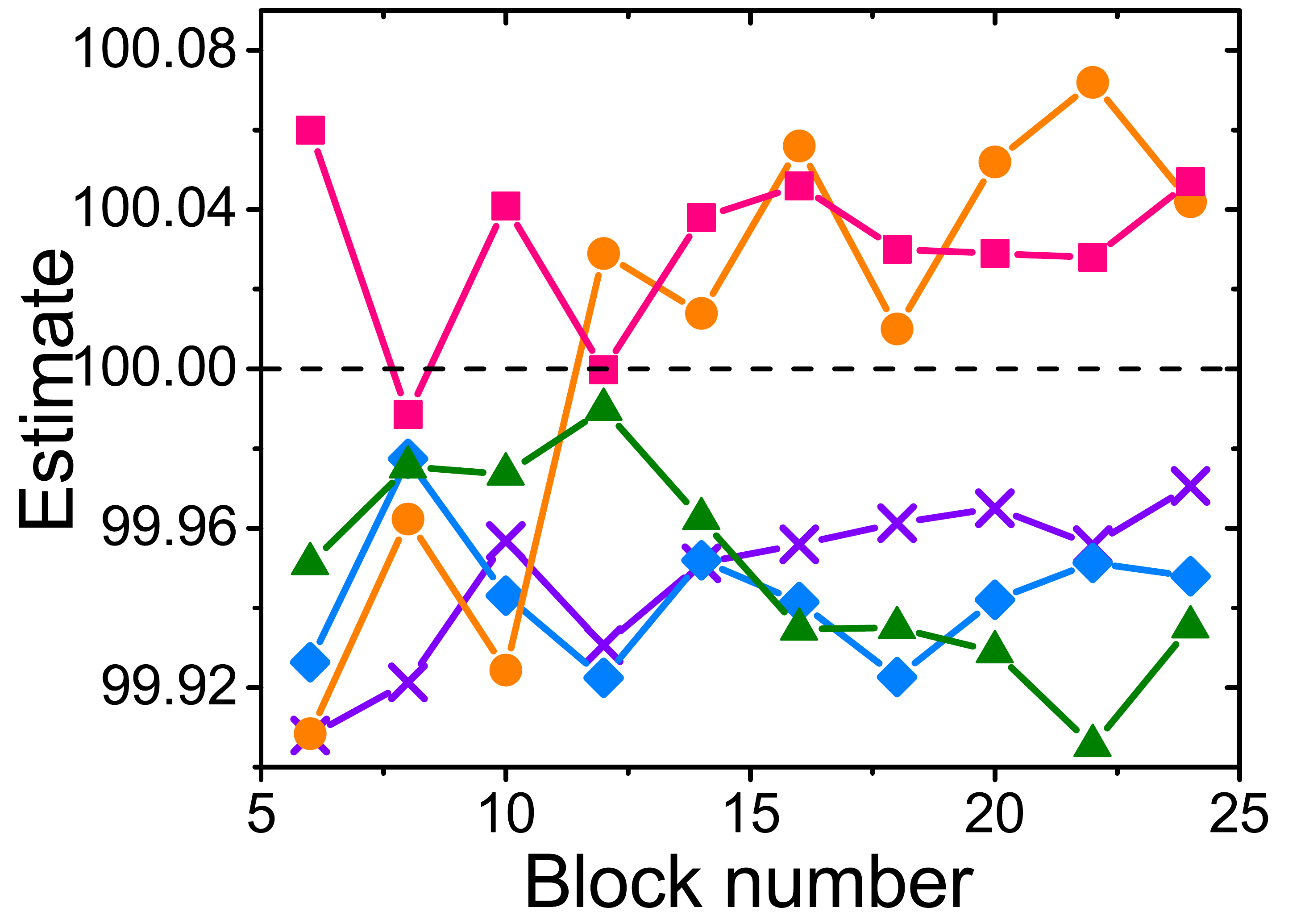}}
\subfigure[Varing data boundaries]{
    \label{fig:dataBoundsVary}
    \includegraphics[width = 4cm]{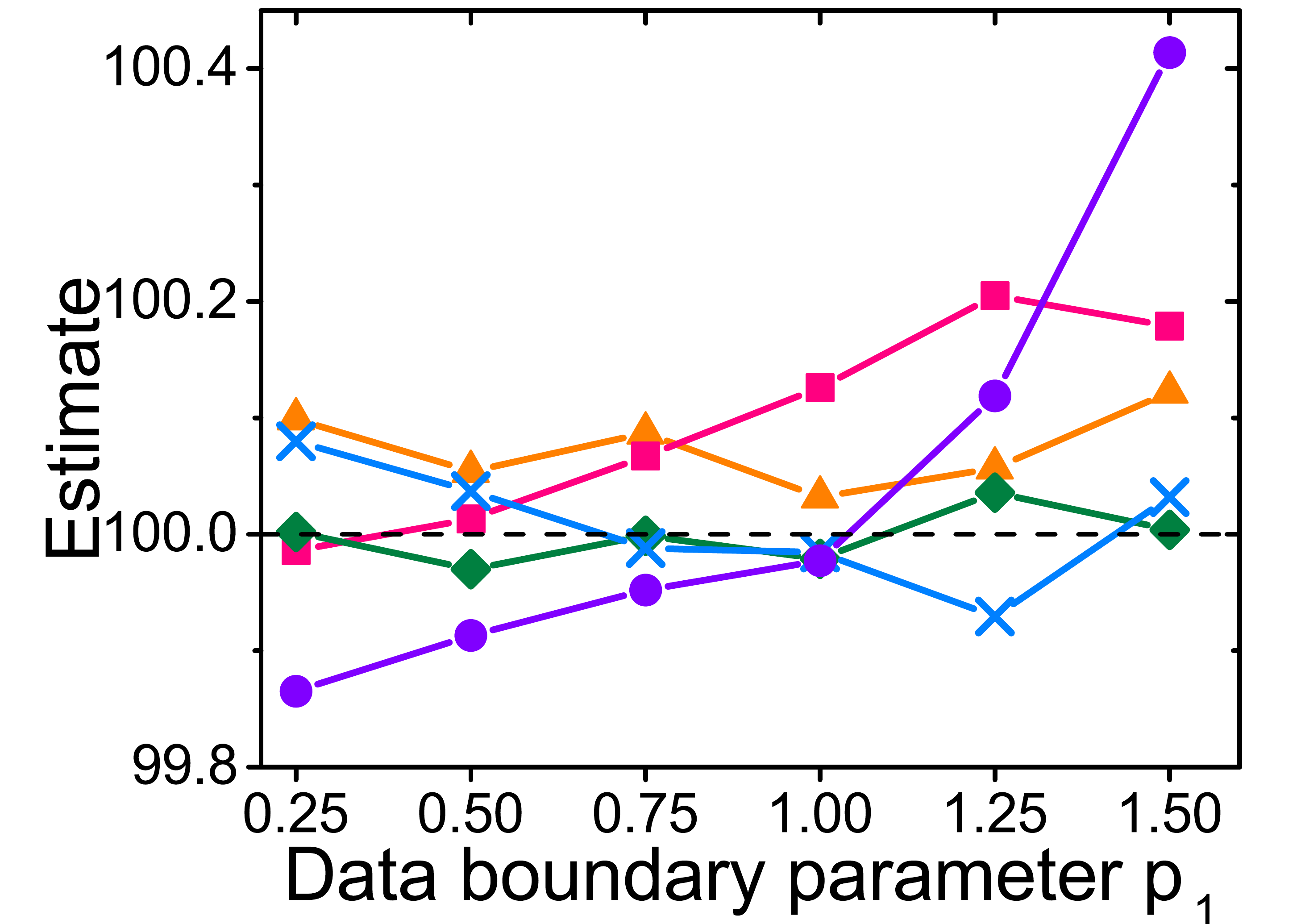}}
    \vspace{-0.5em}
\caption{Impacts of parameters. Five data sets are used. Each line stands for a run. }
\vspace{-1em}
\end{figure}

\noindent
\textbf{Platform.} Our experiments were  performed using a Windows PC of 2.60GHz CPU and 4GB RAM.

\noindent
\textbf{Parameters.} The parameters and the default values are as follows: data size $M$=$10^{10}$, block number $b$=10, desired precision $e$=0.1, confidence $\beta$=0.95, step length factor $\lambda$=0.8,
data boundaries factor $p_1$=0.5 and $p_2$=2.0, and the leverage allocating parameter $q$.
Normally, $q$=1. When the deviation of $sketch_0$ exists, $q$ is generated with $q'$. When $dev$$\in$$(0.94, 0.97)$$\bigcup$$(1.03, 1.06)$, $q'$=5. When $dev$$\in$$(1.06, +\infty)$$\bigcup$$[0, 0.94)$, $q'$=10.


Without special explanations, the sampling rate is determined according to the precision $e$, the confidence $\beta$, and the estimated standard deviation $\sigma$. Data are evenly divided into $b$ parts to process the computations, and are pre-processed and saved in $b$ .txt documents to simulate $b$ blocks.

\begin{table*}
\caption{Comparisons of accuracy. Desired precision: 0.1}
\begin{tabular*}{1\textwidth}{@{\extracolsep{\fill}}llllllllllll}
\hline\noalign{\smallskip}
 \textsf{Dataset} &\textsf{1}&\textsf{2}&\textsf{3}&\textsf{4}&\textsf{5}&\textsf{6}
 &\textsf{7}&\textsf{8}&\textsf{9}&\textsf{10}&\textsf{Average}\\ \hline
\textsf{ISLA} &100.003&
100.003&100.058&
100.064&99.9831&
99.9824&99.995&
100.039&100.076&100.092&100.0296
\\
\textsf{MV}&104.049&
103.96&
104.003&
103.991&
103.958&
104.04&
103.989&
103.997&
104.066&
103.983&104.0036\\
\textsf{MVB} &100.558&
100.472&
100.523&
100.485&
100.471&
100.541&
100.511&
100.51&
100.598&
100.481& 100.515\\
\hline\noalign{\smallskip}
\end{tabular*}
\label{tb: accuracy}
\vspace{-1em}
\end{table*}

\subsection{Impacts of Parameters}\label{sec: ImpactsOfParameter}
We tested the impact of the data size, the required precision, the confidence, the number of blocks, and the data boundaries.

\noindent\textbf{Varying Data Size.}\label{sec: varyingDataSize}
We tested the impact of the data size on the aggregation and demonstrate that our approach can return high-quality answers when processes in a large-scale data processing system.
Data of 100M, 1G, 10G, 100G, and 1TB were generated, with the data sizes of $10^8$, $10^{9}$, $10^{10}$, and $10^{11}$, and $10^{12}$, respectively. The generated data are stored in ``.txt'' files, where each line records a data point. While reading a line, data are handled directly.
We divide the data set into 10 blocks to compute a partial answer. Then the partial answers are collected to generate the final result.
The answers returned by the data of 100M, 1G, 10G, 100G, and 1TB 
are 99.9927, 99.9999, 100.0119, and 100.0035, and 100.0004,
respectively, all satisfying the precision requirement 0.1, indicating that our approach can return high quality answers when dealing with big data (Efficiency of our approach on big data system is evaluated in Section~\ref{sec: efficiency}).
Meanwhile, the returned answers are similar, revealing that the data size has hardly any influence on the aggregation answers.
Actually, according to Section~\ref{sec: sampleRatioEstimation}, the sample size $m$ is only related to $\sigma$, $e$, and $\beta$, suggesting that the data size has hardly any influence on the aggregation answers.


\noindent\textbf{Varying Precision.}
We tested the changing trends in the  aggregation answers with the change in the  desired precision $e$.
We varied the precisions from 0.025 to 0.2. The experimental results are in Fig.~\ref{fig:precisionVary}, which shows that with the increase in the precision, the aggregation results show a trend of divergence.
This indicates that while the precision requirement is relaxed, the accuracy decreases, since the sampling rate is inversely proportional to the value of the desired precision according to Eq.~(\ref{eq: sampleratio}), and lower precision requirements lead to a smaller sampling rate, which produces a decreased precision.

\begin{table*}
\caption{Comparisons of modulation abilities. Desired precision: 0.1}
\vspace{-0.2em}
\begin{tabular*}{1\textwidth}{@{\extracolsep{\fill}}llllllllllll} \hline\noalign{\smallskip}
{\textsf{Partial}}&\textsf{1}&\textsf{2}&\textsf{3}&\textsf{4}&\textsf{5}&
\textsf{6}&\textsf{7}&\textsf{8}&\textsf{9}&\textsf{10}&\textsf{Average}\\
\hline\noalign{\smallskip}
\textsf{ISLA}
&99.9253&
99.9702&
99.9208&
100.065&
100.036&
99.9432&
100.008&
100.193&
99.9573&
100.016&
100.003\\
\textsf{MV}&104.067&
103.949&
104.082&
104.082&
103.987&
104.028&
103.931&
104.117&
104.006&
104.238&
104.049\\
\textsf{MVB} &100.54&
100.499&
100.541&
100.608&
100.496&
100.502&
100.481&
100.654&
100.554&
100.707&100.558\\
\hline\noalign{\smallskip}
\end{tabular*}
\label{tb: modulation}
\vspace{-1.5em}
\end{table*}

\noindent\textbf{Varying Confidence.}
We tested the impact of the confidence $\beta$ using the confidences of 0.8, 0.9, 0.95, 0.98, and 0.99. Experimental results are in Fig.~\ref{fig:confidenceVary}, which shows that with an increase in the confidence, the aggregation answers show a trend of contracting around the accurate value of 100.
This indicates that a higher confidence leads to a better aggregation answer, since the sampling rate increases (according to Eq.~(\ref{eq: sampleratio})), which brings about a more accurate aggregation answer.


\noindent\textbf{Varying Number of Blocks.}
We tested the impact of the number of blocks on aggregation answers.
We generated 5 data sets, varied the number of blocks from 6 to 24, and recorded the aggregation answer of each data set. The results are in
Fig.~\ref{fig:blockNumVary}, which shows that the number of blocks has hardly any influence on the answers. Due to the use of iterations and leverages, high-precision answers are computed according to the actual conditions in each block, leading to the high accuracy of the final aggregation answer.

\noindent\textbf{Varying Data Boundaries.}
We tested the impact of the data boundaries. As discussed in Section~\ref{sec: dataBoundsDefinition}, values out of $(\mu-2\sigma, \mu+2\sigma)$ count for very  limited proportions. Meanwhile, when processing \textsf{AVG} aggregation, they are too far away from the average in distribution, which has limited contributions to the aggregation answers. Thus, we denote such data as outliers in the \textsf{AVG} aggregation and set $p_2$=2. Here we test the impact of $p_1$.
We generated 5 data sets, varying $p_1$ from 0.25 to 1.5, and recorded the aggregation answer of each data set. The results are shown in
Fig.~\ref{fig:dataBoundsVary}.

Fig.~\ref{fig:dataBoundsVary} shows that when $p_1$ is 0.5 or 0.75, \textsf{ISLA} works well. In this condition, the \textsf{S} and \textsf{L} data contain the most featured parts of the normal distributions. Based on such a $p_1$, the \textsf{S} and \textsf{L} data could well predict the distributions.
When $p_1$=0.25, compared with the former condition, more samples are defined as \textsf{S} or \textsf{L} and involved in computing. However, the results are worse than the former condition, since in this condition, the leverages are assigned to more data, leading to stronger leverage effects, which slightly decreases the accuracy.
When $p_1$ gets large, $e.g.$, 1.25 or 1.5, the aggregation answers show a trend of divergence, indicating a low accuracy.
In this condition, $p_1$ is much closer to $p_2$, and the \textsf{S} and \textsf{L} data could not well predict the distributions due to their containing limited features of the distributions.
Besides, fewer samples are used in the computation, which also decreases the accuracy.
In conclusion, we suggest $p_1$ to be 0.5 or 0.75.

\begin{table}
\center
\caption{Evaluation with US and STS. Desired precision: 0.5}
\scriptsize
\begin{tabular}{lllllll}
\hline\noalign{\smallskip}
 \textsf{Data set} &1&2&3&4&5\\
 \hline\noalign{\smallskip}
\textsf{ISLA}
&100.158
&99.8936
&100.136
&99.8917
&100.178\\
\textsf{US}
&99.6591
&99.8918
&99.8675
&99.7068
&99.8371\\
\textsf{STS}
&99.7996
&100.084
&100.261
&99.7332
&99.1607\\
\hline\noalign{\smallskip}
\end{tabular}
\label{tb: levaffect}
\vspace{-2em}
\end{table}

\vspace{-0.5em}
\subsection{Evaluation with Uniform Sampling and Stratified Sampling}\label{sec: experiment4Lev}
\vspace{-0.2em}

We evaluate our approach with uniform sampling and stratified sampling (\textsf{US} and \textsf{STS}). To intuitively show the leverage effects, we set the sampling rate of US as the required sampling rate r, and reduced the sampling rate to r/3 for \textsf{ISLA}. That is, we use only 1/3 of the required sample size, and choose the data in \textsf{S} and \textsf{L} regions to estimate the aggregation answer.
For the convenience of observation, we set the desired precision e to 0.5. 
We generated 5 data sets to conduct the experiments (\textsf{Data set} 1-5), and the results are shown in Table~\ref{tb: levaffect}. The result shows that although \textsf{ISLA} use fewer samples than the \textsf{US} and the \textsf{STS}, all the aggregation answers meet the precision requirement. Moreover, in most of the time, the qualities of the answers calculated by \textsf{ISLA} are even better. This is because our approach introduces leverages to reflect the individual differences in the samples. Due to the leverage effects, our approach achieves high-quality answers with even 1/3 of the required sample size.

\vspace{-0.2em}
\subsection{Evaluation with the Measure-biased Technique}
We compared \textsf{ISLA} with the measure-biased technique in the sample+seek framework~\cite{sampleseek2016}.
The measure-biased technique processes \textsf{SUM} aggregation with the off-line samples, where each data $a$ is picked with the probability proportional to its value:
\vspace{-0.5em}
\begin{equation}\label{eq: measure-biased}
Pr(a)=\frac{a}{\sum_{a'\in A}a'}\sim a.
\vspace{-0.8em}
\end{equation}

Considering that \textsf{AVG} can be computed by dividing the \textsf{SUM} by \textsf{COUNT}, where larger values contribute more to the \textsf{SUM} aggregation answers, we use Eq.~(\ref{eq: measure-biased}) to re-weight the samples in the \textsf{AVG} aggregation.
We also consider the measure-biased technique together with the data division criteria in this paper, and propose another kind of probabilities.
\vspace{-0.2em}
\begin{enumerate}
\item
\underline{\emph{Probabilities on values.}} Probabilities are directly computed with Eq.~(\ref{eq: measure-biased}), proportional to values.

\item \underline{\emph{Probabilities on values and boundaries.}} Probabilities are generated based on the values and data boundaries.
\end{enumerate}
\vspace{-0.2em}
%
%
%

For the second kind of probability, the data are divided into regions according to the data boundaries. Similar to the leverages in Section~\ref{sec: leveragestrategy}, the sum of the probabilities in a specified region is proportional to the number of samples in it. Meanwhile, for samples in a certain region, their probabilities are proportional to their values.
For example, we assume 5 samples are picked with a sum of 100. Two samples, 30 and 35, fall in the region of \textsf{L}. Fors sample 30, its first kind of probability is $\frac{30}{100}$, while its second kind of probability is computed as $\frac{2}{5}$$\times$$100$$\times$$\frac{30}{30+35}$.

In the experiments, we compared \textsf{ISLA} with two measure-biased approaches, the measure-biased approach with probabilities on values (\textsf{MV}), and the measure-biased approach with probabilities on values and boundaries (\textsf{MVB}), to evaluate the accuracy, the modulation effects, and the efficiency of our approach.
%

\noindent\textbf{Accuracy.}\label{sec: accuracyExperiment}
We compared \textsf{ISLA}, \textsf{MV}, and \textsf{MVB} for the accuracy and generated 10 data sets (\underline{\textsf{Dataset 1-10}} in Table~\ref{tb: accuracy}) to run algorithms. The experimental results are shown in Table~\ref{tb: accuracy}.

The average results returned by \textsf{ISLA}, \textsf{MV}, and \textsf{MVB} are 
100.0296,
104.0036, and 100.515, respectively. Only answers calculated by \textsf{ISLA} are satisfied with the desired precision of 0.1.
Meanwhile, detailed answers in Table~\ref{tb: accuracy} indicate that \textsf{ISLA} returns the most robust and high-quality answers when compared with \textsf{MV} and \textsf{MVB}.

\noindent\textbf{Modulation abilities.}
We compared the modulation abilities of \textsf{ISLA}, \textsf{MV}, and \textsf{MVB} to evaluate whether \textsf{ISLA} could properly modulate the sketch estimator in the direction of $\mu$.
We choose the first set of experiments (\underline{\textsf{Dataset 1}}) in Table~\ref{tb: accuracy}, and recorded the partial answers (\underline{\textsf{Partial 1-10}} in Table~\ref{tb: modulation}) to study the modulation process in each block to verify whether \textsf{ISLA} returns better partial results than \textsf{MV} and \textsf{MVB}.
We recorded $sketch_0$, which is $99.676$, and compare $sketch_0$ with the partial results to see whether $sketch_0$ can be properly modulated in each block.
The final answers returned by \textsf{ISLA}, \textsf{MV}, and \textsf{MVB} are 100.003, 104.049, 100.558, respectively.
The experimental results are recorded in Table~\ref{tb: modulation}.
Table~\ref{tb: modulation} shows that the partial results returned by \textsf{ISLA}, with an average of 100.003, are much better, indicating the good modulation abilities of \textsf{ISLA}. Partial results returned by \textsf{MV} and \textsf{MVB} are about 104 and 100.5, respectively, which are both outside of the confidence interval ($sketch_0-0.1, sketch_0+0.1)$, leading to poorer answers.


\vspace{-0.2em}
\subsection{Experiments on \emph{Non-i.i.d.} Distributions}
In Section~\ref{sec: extension-noniid} we extend our approach to process the \textsf{AVG} aggregation on \emph{non-i.i.d.} distributions, and we now test the performance.
We generated 5 blocks in different normal distributions $N(\mu, \sigma^2)$: $N(100, 20^2)$, $N(50, 10^2)$, $N(80, 30^2)$, $N(150, 60^2)$, and $N(120$, $40^2)$, with the data size of $10^8$ in each block.
The accurate answer is 100, which is calculated by dividing the sum of the accurate averages of each block.
The desired precision $e$ is set to 0.5.
We conducted the experiments 5 times. The aggregation answers are 99.8538, 100.066, 100.194, 100.321, and 99.8333, respectively. All the results satisfy the desired precision, indicating that our approach has a good performance for \emph{non-i.i.d.} distributions.
\subsection{Other Distributions}\label{sec: expotherDistribution}
We experimentally show that our method is also suitable for other kinds of distributions.
Similar to the comparison experiments above, we compare \textsf{ISLA} with \textsf{MV} and \textsf{MVB}. Without a specific explanation, the parameters are set to default values.

\begin{table}
\center
\caption{Experiments on exponential distributions}
\begin{tabular}{lllll}
\hline\noalign{\smallskip}
\textsf{$\gamma$}&
\textsf{0.05}&
\textsf{0.1}&
\textsf{0.15}&
\textsf{0.2}\\
\hline\noalign{\smallskip}
\textsf{Accurate}
&20
&10
&6.67
&5\\
\textsf{ISLA}
&19.8713
&9.53488
&6.32677
&4.60377\\
\textsf{MV}
&39.7174
&20.2711
&13.2486
&10.3369\\
\textsf{MVB}
&21.8042
&11.0635
&7.30495
&5.49333\\
\hline\noalign{\smallskip}
\end{tabular}
\label{tb: expExperiment}
\vspace{-1em}
\end{table}

\begin{table}
\center
\caption{Experiments on uniform distributions}
\begin{tabular}{llllll} \hline\noalign{\smallskip}
\textsf{Dataset} &\textsf{1}&\textsf{2}&\textsf{3}&\textsf{4}&\textsf{5}\\
\hline\noalign{\smallskip}
\textsf{ISLA} &99.7658&99.5098&99.5627&99.7011&99.8016\\
\textsf{MV}&132.031&132.046&131.932&132.12&132.06\\
\textsf{MVB} &93.5209&92.8587&93.3415&93.7927&95.3857\\
\hline\noalign{\smallskip}
\end{tabular}
\vspace{-1.5em}
\label{tab: unif}
\end{table}

\noindent\textbf{Exponential Distributions.} We designed our approach based on the symmetry of normal distributions, and we want to test its performance on asymmetrical distributions. Thus, we choose the exponential distribution 
with the probability density function
$f(x)$=$\gamma$$e^{-\gamma x} (x$$>$$0)$,
where the accurate average is $1/\gamma$.
Note that when $\gamma$ increases, $1/\gamma$ decreases, for the convenience of observation and comparison, we vary $\gamma$ from 0.05 to 0.2. We record answers calculated with \textsf{ISLA}, \textsf{MV}, and \textsf{MVB}. The accurate averages $1/\gamma$ are also included as a comparison. The results are recorded in Table~\ref{tb: expExperiment}.
Table~\ref{tb: expExperiment} shows that answers returned by \textsf{ISLA} even outperform the competitors, and \textsf{ISLA} is capable for \textsf{AVG} aggregation on exponential distributions. 



\noindent\textbf{Uniform Distributions.}
We generated random data uniformly from the range [1, 199] 5 times (\underline{\textsf{Dataset 1-5}} in Table~\ref{tab: unif}) to conduct experiments on uniform distributions to compare the robustness of \textsf{ISLA} with \textsf{MV} and \textsf{MVB}. The accurate average is 100. The results are shown in Table~\ref{tab: unif}.
Table~\ref{tab: unif} shows that answers returned by \textsf{MV} are around 132, and answers returned by \textsf{MVB} are from 92.8 to 94.3.
\textsf{ISLA} obviously returns much better results, varying from 99.5 to 99.85, indicating that \textsf{ISLA} is much more robust than the competitors.
Even so, when dealing with such kinds of extreme cases, the accuracy decreases, with the desired precision unsatisfied.
This is because the uniform distribution is an extreme condition of normal distributions with a very large standard deviation $\sigma$, leading to a loss of precision.
An improvement can be performed by increasing the overall sampling rate accordingly, and this is left for further research.


%

\subsection{Evaluation on Large-scale Data Processing}\label{sec: efficiency}
Since it is impractical to compute an accurate answer on big data, in Section~\ref{sec: varyingDataSize}, we generate data of 100M, 1G, 10G, 100G, and 1T to run the algorithm and compare the answers with the accurate average 100.
The results show that our approach can return high-quality answers when dealing with big data.
We also use real-world data of an appropriate data size to conduct experiments in Section~\ref{sec: realDataExp}. And the accurate average could be obtained through a full scan of the data.

To validate the high efficiency of our approach, in this section, we used TCP-H benchmark to produce 100 GB data~\cite{realdataTPCH}, and chose the column LINEITEM to process \textsf{AVG} aggregation. The data size is 600 million.
We compared \textsf{ISLA} with \textsf{MV}, \textsf{MVB}, \textsf{US}, and \textsf{STS} on efficiency measured by run time. Each algorithm is run for 20 times to get a total run time.
The time required to run \textsf{ISLA}, \textsf{MV}, \textsf{MVB}, \textsf{US}, and \textsf{STS} for 20 times are 31979ms, 61718ms, 70584ms, 25989ms, and
84294ms, respectively. Results show that \textsf{ISLA} is more efficient than \textsf{MV}, \textsf{MVB}, and \textsf{STS}, and is a little less efficient than \textsf{US}.
However, according to the former experiments, results returned by \textsf{ISLA} achieve a much higher precision. Furthermore, comparing with \textsf{US} and \textsf{STS}, \textsf{ISLA} returns better results while using a smaller sample size.
In conclusion, \textsf{ISLA} is suitable to deal with big data.

\vspace{-0.2em}
\subsection{Results on Real Data}\label{sec: realDataExp}
We choose the real-world data of an appropriate data size and compare the estimated results with the accurate value to evaluate the quality of answers, for it usually costs more than 3 hours when dealing with big data, which is impractical. With using an appropriate data size, the accurate average could be obtained through a full scan of the data. We conducted experiments on 2 data sets, the salary data~\cite{realdata}, and TLC Trip~\cite{realdataTaxi}. 
We compare our approach with \textsf{MV}, \textsf{MVB}, \textsf{US}, and \textsf{STS}.

\noindent\textbf{Salary data.} The salary data was extracted from the 1994 and 1995 population surveys conducted by the U.S. Census Bureau. The data size is 299285, with an accurate average of 1740.38. 
Data are divided into 10 blocks, and 1000 samples are picked uniformly and randomly from these blocks to generate $sketch_0$. To effectively validate our approach, we set the sample size of \textsf{MV}, \textsf{MVB}, \textsf{US}, and \textsf{STS} to 20000, and set the sample size of \textsf{ISLA} to only 10000.
Answers computed by \textsf{ISLA}, \textsf{MV}, \textsf{MVB}, \textsf{US}, and \textsf{STS} are 1731.48, 2326.78, 1798.78, 1742.79, and 1740.37, respectively.


\noindent\textbf{TLC Trip data.} We used the yellow car data of January, 2016 and choose the column of ``trip\_distance'' to conduct experiments to evaluate the results returned by \textsf{ISLA} with \textsf{MV}, \textsf{MVB}, \textsf{US}, and \textsf{STS}. To directly see the differences of the results, each data values are multiplied by 1000. The data size is 10906858, with an accurate average of 4648.2. While handling the data, we found that the data set is highly-skewed. The too big values and the too small values are highly clustered. So we wonder how our approach performs on such data set.
Answers returned by \textsf{ISLA}, \textsf{MV}, \textsf{MVB}, \textsf{US}, and \textsf{STS} are 4515.73, 7426.37, 3298.09, 2908.53, and 4289.08, respectively.

Results of these two experiments above show that \textsf{ISLA} has much better performance than \textsf{MV}, \textsf{MVB}, \textsf{US}, and \textsf{STS}, indicating \textsf{ISLA} returns high-quality answers.

\section{Related Work}\label{sec: related}
We briefly survey work related to this paper, including sampling strategies, and approximate query processing.
\vspace{-0.5em}
\subsection{Sampling strategies}
\noindent\textbf{Bi-level sampling and block-level sampling.}
Bi-level sampling~\cite{Haas2004A} combines the row-level and page-level sampling to control a tradeoff between efficiency and accuracy. The local variance of blocks is considered to generate different sampling rates, and uniform blocks are sampled less than blocks with large variances.
Block-level sampling~\cite{Chaudhuri2004Effective}
uses fewer blocks and larger sample sizes in these blocks and accesses for the same sample size to reduce IO.
In our approach, to increase the accuracy, we mainly consider the individual differences in the samples instead of the local variance of the blocks.
For the ease of discussion, we assume the data to be identically distributed on the blocks. We also consider the local variance and extend our approach to deal with \emph{non-i.i.d.} data, as discussed in Section~\ref{sec: extension-noniid}.

\noindent
\textbf{Leverage-based sampling.}
The leverage-based sampling technique~\cite{Ma2013A} picks biased samples in the leverage-based probabilities. Leverages are generated in a single way, calculated with the data value as well as all the data. Meanwhile, all the samples are involved in aggregation. Thus, the influence of outliers could not be eliminated.
In our approach, samples are uniformly picked and then re-weighted.
We consider the nature of data and divide the data into regions according to their features and then assign various leverages to handle them differently.
Moreover, we completely avoid the influence of outliers by selecting the \textsf{S} and \textsf{L} samples into the computation, which could reflect the features of the distributions very well, leading to a high accuracy and efficiency when dealing with big data.

\noindent
\textbf{Error-bounded stratified sampling.}
Error-bounded stratified sampling~\cite{error2014vldb} focuses on sparse data and divides data into regions and samples them differently to reduces the sample size, while our approach focuses on the most common distributions (normal distributions, or similar normal distributions), picks uniform samples and handles them differently.
\subsection{AQP (Approximate Query Processing) and AQP++}
\noindent
\textbf{Off-line processing.}
This technique prepares the samples or summaries in advance in order to  execute queries~\cite{Sidirourgos2011SciBORQ}~\cite{AgarwalBlinkDB}~\cite{sampleseek2016}~\cite{AQUA1999SIG}.
SciBORQ~\cite{Sidirourgos2011SciBORQ} picks samples based on previous results, and the bias of samples is considered, where tuples from the areas of interest are more likely to be picked.
BlinkDB~\cite{AgarwalBlinkDB} generates multi-dimensional samples and use dynamic sample selection strategies to provide fast responses.
Sample+seek~\cite{sampleseek2016} classifies queries into large queries and small queries according to the hardness of being answered. A measure-biased technique is proposed to process small queries with the off-line samples, and indexes are provided to aid the off-line samples for large queries.
Aqua~\cite{AQUA1999SIG} computes summary synopses in advance, and approximate answers are generated by rewriting and executing queries over the synopses.
These techniques may either require elaborate off-line processes in advance or depend much on the previous queries, which can be unavailable and less flexible when dealing with new data sets.

Our approach does not require much off-line processing. It just draws a small sample set to get an estimated average as one of the two estimators for the computation. To increase the quality of the results, we also use it to generate the data boundaries. For the ease of illustration, we define this computation into the Pre-estimation Module. Actually, due to the importance of this value in the iteration, it should belong to a part of our iterative algorithm. Meanwhile, to avoid great changes of the data (such as, someone inserts huge amounts of rows into the table), we suggest not draw the pilot samples much ahead of the computation.

\noindent
\textbf{Online aggregation.}
Online aggregation \cite{online1997}\cite{Wu2010Distributed} allows users to observe the promotion of answers and cut off the computation when a desired answer is obtained. 
However, it requires users to observe all the time, leading to a poor user experience. 
Moreover, samples are treated identically without considering their individual differences, leading to a loss of accuracy to some degree. Our approach considers individual differences in the data and directly returns answers with confidence assurance which does not require observation.

\noindent
\textbf{AQP++.}
AQP++ connects aggregate pre-computation and AQP together for interactive analytics~\cite{AQPPP}. Given a table D, a column C, a sample set S, and a pre query $pre$: select $f(A)$ from $D$ where $x_0 \leq y_0$, to execute query $q$: select $f(A)$ from $D$ where $x\leq y$, it uses the equation $q(D) - pre(D) \approx \hat{q}(S) - \hat{pre}(S)$. AQP++ is mainly for range queries, while our approach focuses on computing high-quality answers of a given data set. Considering of this, we could combine our approach with AQP++ to process high-quality aggregation computation for range queries.

\vspace{-0.2em}
\section{Conclusion}\label{sec: conclusion}
In this paper, we propose an effective approach to calculate high-accuracy aggregation answers with only a small portion of the data.
To increase the accuracy, the leverage mechanism is introduced to reflect individual differences among the samples.
Two estimators, the \emph{l-estimator} and the sketch estimator, are built, which are gradually improved based on their relations until the difference is below a threshold.
This process is done over multiple iterations, using multiple modulation strategies according to the actual conditions of the data. Based on the iterations, proper aggregation answers are obtained.
Without requiring storage for the samples, our algorithm achieves a high efficiency and works well when dealing with big data.
More analysis and experiments for extensions, such as extreme value aggregation, are left for future work.

\vspace{-0.2em}

\section{appendices}
%

\subsection{Proof of Theorem~\ref{th: function}}
The leverage-based answer $\hat{\mu}$ is generated as follows.

\noindent
\textbf{1. Leverage assignment.} Initially, the original leverages are assigned to the \textsf{S} and \textsf{L} samples.
For $\forall$$x$$\in$$X$, the original leverage is $1-\frac{x^2}{\sum x_i^2 + \sum y_i^2}$; for $\forall$$y$$\in$$Y$, the original leverage is $\frac{y^2}{\sum x_i^2 + \sum y_i^2}$.

\noindent
\textbf{2. Normalization factor calculation.} We get the normalization factors $fac$ for \textsf{S} and \textsf{L},  respectively, by dividing the sum of the leverage scores by the theoretical sum of the leverages.
For \textsf{S} samples, $fac_x$ = $(u+\frac{v}{q})$ $(1-\frac{\sum x_i^2}{u(\sum x_i^2 + \sum y_i^2)})$; for \textsf{L} samples, $fac_y$ = $(q\frac{u}{v}+1)$ $(\frac{\sum y_i^2}{\sum x_i^2 + \sum y_i^2})$.

\noindent
\textbf{3. Leverages normalization.}
The leverages of the \textsf{S} and \textsf{L} samples are calculated by dividing the original leverages by $fac_x$ and $fac_y$.
For $\forall$ $x$ $\in$ $X$, $lev_x$ = $\frac{1-{x^2}}{fac_x(\sum x_i^2 + \sum y_i^2)}$;
for $\forall$ $y$ $\in$ $Y$, $lev_y$ = $\frac{y^2}{fac_y({\sum x_i^2 + \sum y_i^2})}$.

\noindent
\textbf{4. Re-weighted probability generation.} The probabilities of the samples are generated according to Eq.~(\ref{eq: probability}), with the uniform sampling probability equal to
$\frac{1}{u+v}$. For $\forall$$x$$\in$$X$, $prob_x$ = $\alpha lev_x$ + $\frac{1-\alpha}{u+v}$; for $\forall$ $y$ $\in$ $Y$, $prob_y$ = $\alpha  lev_y$ + $\frac{1-\alpha}{u+v}$.

\noindent
\textbf{5. L-estimator generation.}
The value of the $l$-$estimator$, $\hat{\mu}$, is computed as $\hat{\mu}$=$\sum x \cdot prob_{x}$+$\sum y \cdot prob_{y}.$
After putting terms related to $\alpha$ together and accumulating the coefficients, the coefficient of $\alpha$ is obtained, and the function
of $\alpha$ is derived: $\hat{\mu}$=$f(\alpha)$=$k \alpha + c$,
 where $k=(\frac
 {(\sum x_i^2+\sum y_i^2)\sum x_i-\sum x_i^3}
 {(1+\frac{v}{qu})(u(\sum x_j^2 + \sum y_j^2)-\sum x_i^2)}$ + $\frac{v\sum y_i^3}{(qu+v)\sum y_j^2})-\frac{u+v}{\sum x_i + \sum y_i}$, and $c$ = $\frac{u+v}{\sum x_i + \sum y_i}$.

\end{document}